\documentclass[journal]{IEEEtran}
\usepackage{graphicx}
\usepackage{subfigure}
\usepackage{amsmath,amssymb,amsthm}
\usepackage{bbm}
\usepackage{cite}
\usepackage[boxed,ruled,lined]{algorithm2e}
\usepackage{stfloats}

\def\bE{\mathbb{E}}

\def\bR{\mathbb{R}}

\def\cH{\mathcal{H}}

\def\cR{\mathcal{R}}

\def\cT{\mathcal{T}}

\def\cX{\mathcal{X}}

\def \qed {\hfill \vrule height6pt width 6pt depth 0pt}
\def\bee{\begin{equation}}
\def\ene{\end{equation}}
\def\beq{\begin{eqnarray}}
\def\enq{\end{eqnarray}}

\newtheorem{defi}{Definition}
\newtheorem{lem}{Lemma}
\newtheorem{thm}{Theorem}

\newtheorem{rem}{Remark}
\newtheorem{exam}{Example}

\begin{document}

\title{Likelihood Ratio Based Scheduler for Secure Detection in Cyber Physical Systems\thanks{This work was partially
supported by National Natural Science Foundation of China under grant NSFC 61273233.}}

\author{Jian-Ya Ding,
        Keyou You,~\IEEEmembership{Member,~IEEE},  Shiji Song,~\IEEEmembership{Member,~IEEE}, Cheng Wu
\thanks{The authors are with the Department of Automation, Tsinghua University, Beijing, 100084, China. (e-mail:~ding-jy12@mails.tsinghua.edu.cn, \{youky, shijis, wuc\}@mail.tsinghua.edu.cn)}}
\maketitle

\begin{abstract}
This paper is concerned with a binary detection problem over a non-secure network. To satisfy the communication rate constraint and against possible cyber attacks, which are modeled as deceptive signals injected to the network, a likelihood ratio based (LRB) scheduler is designed in the sensor side to smartly select sensor measurements for transmission. By exploring the scheduler,  some sensor measurements are successfully retrieved  from the attacked data at the decision center. We show that even under a moderate communication rate constraint of secure networks, an optimal LRB scheduler can achieve a comparable asymptotic detection performance to the standard N-P test using the full set of measurements,  and is strictly better than the random scheduler. For non-secure networks, the LRB scheduler can also maintain the detection functionality but suffers graceful performance degradation under different attack intensities. Finally, we perform simulations to validate our theoretical results.

\end{abstract}

\IEEEpeerreviewmaketitle

\section{Introduction}

Cyber-Physical Systems (CPS) are intelligent systems which are typically  networked systems for interactions between physical processes and computing components. Computation and communication capabilities are deeply embedded in CPS to help detect, monitor and control physical entities in physical system remotely in a cooperative way. Examples of CPS include intelligent transportation systems, environmental monitoring networks, power distribution networks, healthcare monitoring systems, wireless sensor networks (WSNs) and etc. These applications are generally safety-critical since system failures due to malicious cyber attacks can result in disastrous breakdown to infrastructures' functionality and irreversible harm to human safety.

%
Distributed detection \cite{tenney1981detection} in WSNs plays an important role in network monitoring. A typical distributed detection sensing system consists of a decision center, a wireless communication link and a group of sensor nodes, which are deployed in the interested area for data collection.  WSNs are usually deployed in public or adversarial environment over non-secure channels \cite{alrajeh2013intrusion,bokareva2006wireless}, under various potential cyber attacks such as jamming \cite{li2013jamming}, selective forwarding \cite{yu2006detecting}, message manipulation\cite{he2013sats}, false data injection \cite{zhu2004interleaved} and routing attacks \cite{lu2005applying}. These attacks can induce deceptive false alarms in detection and lead to misjudging of the system state changes. For instance, an intruder may fabricate an equipment breakdown event in a process industry by sending deliberately forged signals to the detection center. Considerable economic loss will occur if the ongoing production process are terminated due to the non-existing equipment breakdown. To protect the system, secure communication mechanisms and technologies are highly desired.




Previous research concerning security issues in WSNs mainly focus on information authentication \cite{perrig2002spins,lu2012becan}, key encryption \cite{du2005pairwise,boneh2003identity}, data aggregation \cite{ozdemir2009secure,he2007pda}, and secure routing \cite{zhou1999securing,karlof2003secure}. These security measures are generally designed to defend against malicious attacks by preventing deceptive data sent or received by the sensors, controllers and actuators,  while the survivability under attacks is rarely discussed.   In practice, however, the defense mechanisms can often fail due to various unpredictable problems such as human errors, design defects, incorrect operations, and device misconfigurations. This indicates that attacks may not be resolved even when defense technology is employed in the system. As Cardenas et al  pointed out in \cite{cardenas2008secure}, it is highly desirable that the system proceed to function well or degrade gracefully even under attacks.


In this work,  we consider a detection framework in Fig. \ref{fig1} where the system and the remote tester are connected by a non-secure network. In particular, an attacker will inject a deceptive signal $q_i$ into the channel.  If the measurement $y_i$ is directly sent to the tester, the attacker renders it difficult for the tester to distinguish from the deceptive signal to make a correct decision.
\begin{figure}
  \centering
  \includegraphics[width=3.40in]{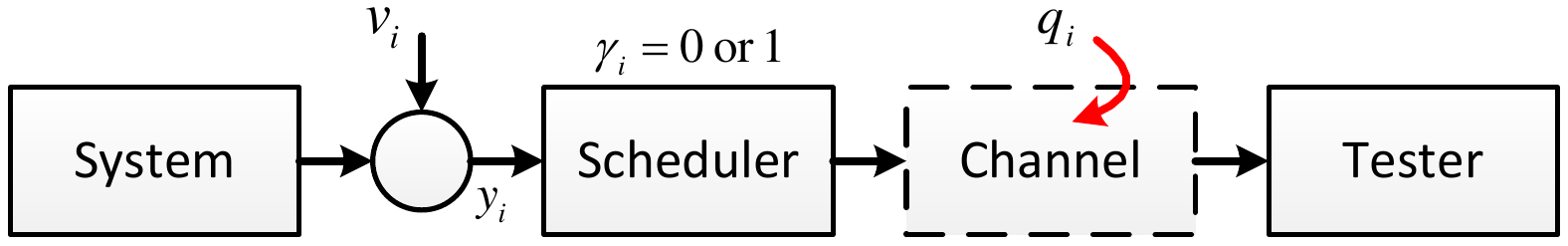}
    \caption{ Distributed detection framework under attack}
  \label{fig1}
\end{figure}
To solve it, we propose a communication scheduler in the transmitter side whose state is denoted by $\gamma_{i}$. For example, $\gamma_{i}=0$ corresponds to that $y_{i}$ will not be transmitted to the tester while $\gamma_{i}=1$ indicates that $y_{i}$  will be sent to the tester.  The purpose of such a scheduler is twofold. Firstly, the binary-valued $\gamma_i$ of the scheduler can be securely encrypted, and provides critical information for the tester to survive the ongoing cyber attack. Secondly,  it helps to select ``useful'' measurements for transmission, and discard the rest of ``non-useful'' ones. From this perspective, the scheduler reduces the energy consumption for sensors as only a subset of the sensor measurements will be transmitted to the tester.



Observe that the likelihood ratio of sensor measurements is key to the celebrated Neyman-Pearson (N-P) test, which is known to be an optimal test for the binary hypothesis testing problem under a full set of measurements \cite{casella2001statistical}. This suggests that the likelihood ratio might be a good candidate to quantify the importance of a sensor measurement. To pursue this idea, we propose a  likelihood ratio based (LRB) communication scheduler for a binary hypothesis testing problem. Specifically, only the sensor measurement with its likelihood ratio far from one will be useful for the detection problem and sent to the tester while the rest are deemed as non-useful.

Note that the tester will receive the transmitted measurements with the securely encrypted scheduler state $\gamma_i$ and possible deceptive signals. By exploring their relationships, the tester can extract certain critical information for detection. Without cyber attacks, {\em scheduled transmission rate} is introduced in \cite{you2013asymptotically} to quantify the reduction of communication cost by the scheduler, which is defined as the ratio of the number of transmitted measurements to the total number of sensor measurements in the average sense. Intuitively, the larger the scheduled transmitted rate, the larger the number of sensor measurements tends to be transmitted in a fixed time interval. Due to possible cyber attacks, a larger transmission rate does not always result in a better detection performance as it may also reduce the chance for the tester to discriminate the transmitted true measurements from the deceptive signals. We are therefore interested in the design of an optimal LRB scheduler to achieve a good detection performance under an appropriate transmission rate.  Since the scheduler is a nonlinear function of sensor measurements, it is challenging to design an (asymptotically) optimal LRB scheduler for detection.


Under a LRB scheduler, which is to be specified later, we explicitly derive the rejection regions of the null hypothesis respectively under N-P tests. Due to nonlinearity of the scheduler, it is impossible to analytically calculate the probabilities of either Type I or II errors. Thus, we resort to the asymptotic approach  to quantify the effect of the scheduler on the detection performance, and obtain optimal schedulers.  Note that sensor data scheduling  has been reported in \cite{rago1996censoring,appadwedula2005energy,appadwedula2008decentralized,maleki2011energy,ding2014likelihood} for detection problems and in \cite{shi2011optimal,shi2011sensor,wu2013event,you2013asymptotically} for state estimation problems, while none of these works consider security issues. Besides, they focus on establishing an optimal transmission strategy under different scenarios while the effect of sensor scheduling mechanism on the detection performance remains unexploited. In comparison, we conduct an asymptotic analysis as the number of sensor measurements tends to infinity. This enables us to clearly quantify how the scheduled transmission rate and the potential deceptive data degrades the asymptotic detection performance.

The contribution of this work can be summarized as follows. Firstly, inspired by the feature of an N-P test, we propose an LRB communication scheduler to reduce the communication cost for a binary detection problem, which is shown to be an efficient and simple strategy both in theory and simulation.
Secondly, under a scheduled transmission rate constraint, an asymptotically optimal LRB scheduler is designed to achieve the fastest exponential convergence of detection error.
Thirdly, using the scheduler information, we introduce tester signal processing protocols to maintain detection functionality even when deceptive attack signal are injected in the channel. In addition, the degradation of detection performance under a certain attack intensity is explicitly quantified in the asymptotical sense.


 The remainder of this paper is organized as follows. In Section II, we formulate the distributed detection problem and propose the LRB scheduler. In Section III, we explicitly develop the Neyman-Pearson test and conduct asymptotic performance analysis for the LRB scheduler. Section IV introduces signal processing protocols at the tester when cyber attack happens. Asymptotic analysis are conducted to exactly quantify performance degradation of detection under different attack intensities. In Section V, simulation results are included to validate our theoretical results. In Section VI, we draw some concluding remarks and highlight some future research directions.


\section{Problem Formulation}
Distributed detection through WSNs is threatened by potential cyber attacks, but works under friendly communication environment in the majority of time. We are particularly interested in designing a secure detection system that performs well under malicious environment.

\label{section2}
\subsection{Communication Scheduler for Secure Detection}
Consider a scenario that a remote sensor is deployed to monitor a phenomenon of physical interest, which is modeled as a binary state of the nature. The  sensor collects samples from the system and transmits them to a decision center through a non-secure network. In this paper,  the sensing signal is given by
\begin{equation}
y_{i}=\theta+v_{i},i=1,2,\ldots, \label{system}
\end{equation}
where the system state is represented by a binary-valued $\theta\in \{\theta_0, \theta_1\}\subset\bR$, and $v_{i}$ is a white Gaussian noise with variance $\sigma^2$.
It is obvious that it does not lose generality by letting $\theta_{0}<\theta_{1}$.

The task of the decision center is to discriminate between the two states using its received sensing information. To this purpose, a tester is implemented to solve the following hypothesis testing problem:
\bee
\label{prb_test}
\cH_{0}:\{\theta=\theta_{0}\}\quad\quad \text{versus}\quad\quad \cH_{1}:\{\theta=\theta_{1}\}.
\ene

In statistics, $\cH_0$ is called the null hypothesis and $\cH_1$ is the alternative hypothesis. Throughout the paper, all the random varibles/vectors are defined on a common probability space $(\Omega,\mathcal{F},P_\theta)$. In contrast to the conventional hypothesis testing problem \cite{casella2001statistical}, we focus on a distributed framework consisting of a communication scheduler, which is embedded in the sensor side,
a remote tester and a non-secure channel, see Fig. \ref{fig1} for an illustration. In particular, let $\gamma_{i}=1$ indicate that the sensor is triggered to transmit $y_i$ to the tester at sampling time $i$ and the packet containing the information of $y_{i}$ is perfectly delivered to the tester while $\gamma_{i}=0$ means that there is no data transmission between the sensor and the tester. That is, the scheduler will decide whether $y_{i}$ is to be transmitted or not. A natural question is how to design an efficient and easily implementable scheduler to achieve a good detection performance, even under the potential cyber attacks.

As in \cite{you2013kalman}, we do not consider other communication issues such as packet losses, transmission delay, data quantization, channel noise and etc. This implies that if there is no measurement sent from the sensor at time $i$, the tester does not receive anything from the channel and understands that $y_i$ is censored. Hence, the  information received by the tester at time $i$ is given by
\begin{equation}
z_{i}=\{\gamma_{i}y_{i},\gamma_{i}, q_i\}.
\end{equation}

For noisy channels, the binary-valued $\gamma_i$ can be available to the tester by sending merely one bit message from the sensor node. It should be noted that in light of \cite{you2013asymptotically} the main approach of this work can be generalized to address some channel medium effects, although we do not plan to do so for brevity.

Given an arbitrary scheduler $\{\gamma_i\}$, we adopt the so-called {\em scheduled transmission rate}\cite{you2013asymptotically} to roughly evaluate the average communication cost at time $N$ as follows
\begin{equation}\label{schedulerrate}
R_{\theta}=\frac{1}{N}\sum_{i=1}^{N}\mathbb{E}_{\theta}[\gamma_{i}],
\end{equation}
where the mathematical expectation $\bE_\theta[\cdot]$ is taken with respect to the probability measure $P_\theta$.
Intuitively, $R_\theta$ characterizes the frequency of measurement transmission in the average sense. Obviously, a larger scheduled transmission rate corresponds to a higher ratio of the number of transmitted measurements to the total number of sensor measurements.

Under the above detection framework, the objective of this work is as follows.
\begin{enumerate}
 \item Given a scheduled transmission rate constraint, design an efficient and practical scheduler in the sense of achieving a good detection performance.
 \item Design tester signal processing protocols to maintain system functionality when deceptive signals are injected into the system.
\item  Exactly quantify the effect of the scheduler and different attack intensities on the asymptotic detection performance.
\end{enumerate}


\subsection{LRB Scheduler}
Given a full set of sensor measurements $\mathbf{y}=\{y_{1},\ldots,y_{N}\}$,  recall that the rejection region for the null hypothesis of the standard Neyman-Pearson (N-P) test for a binary detection problem is expressed by
\begin{equation}
\{\mathbf{y}:f(\mathbf{y}|\theta_{1})>kf(\mathbf{y}|\theta_{0})\},\label{eq:rej}
\end{equation}
where  $f(\mathbf{y}|\theta_{j})$ denotes the joint probability density function (pdf) of ${\bf y}$ with respect to $\theta\in\{\theta_0, \theta_1\}$ and $k$ is a decision threshold. It is well known that
the N-P test is a uniformly most powerful test for the detection problem (\ref{prb_test}) under the full set of measurements ${\bf y}$.

Since ${y}_i$ is a sequence of independently and identically distributed (i.i.d.) Gaussian random variables,
the rejection region in (\ref{eq:rej}) can be explicitly expressed as
\begin{equation}\label{rejection}\left\{ {{\bf{y}}:\prod\limits_{i = 1}^N {\frac{{{p_{{\theta _1}}}({y_i})}}{{{p_{{\theta _0}}}({y_i})}}}  > k} \right\},
\end{equation}
where $p_{\theta}(y_{i})=\mathcal{N}(y_{i};\theta,\sigma^{2})$ and $\mathcal{N}(y_{i};\theta,\sigma^{2})$ is the pdf of a Gaussian random variable with mean $\theta$ and variance $\sigma^{2}$.

As the rejection decision is mainly decided by the likelihood ratio, we resort to the likelihood ratio function to quantify the importance of a sensor measurement as follows
\begin{equation}
g(y|\theta_{0},\theta_{1})=\frac{p_{\theta_{1}}(y)}{p_{\theta_{0}}(y)}.
\end{equation}

In fact, if $g(y_i|\theta_{0},\theta_{1})$ is far from one, the measurement $y_i$ will lead to a significant change of the likelihood ratio, and eventually affects the tester decision. Thus, we expect that $y_{i}$ associated with a likelihood ratio far from one may contain useful information for detection, and a measurement with its likelihood ratio close to one is regarded as non-useful or uninformative.

Taking advantage of the above observation and noting that $y_i$ is an i.i.d. random sequence, we propose a likelihood ratio based (LRB) scheduler of the following form:
\begin{equation}
\gamma_{i}=\left\{
\begin{array}{rcl}
0,       &      &\text{if  } { {1}/{\alpha}<g(y_{i} | \theta_{0},\theta_{1})<\alpha};\\
1,     &      & \text{otherwise,}
\end{array} \right.  \label{LRB scheduler}
\end{equation}
where the scheduling parameter $\alpha>1$ and is to be designed. Informally speaking, the scheduler helps the sensor to smartly decide whether a sensor measurement is useful or not.  To reduce the number of sensor communication, only useful measurements will be transmitted to the tester through a non-secure network and discard the rest of measurements.

The LRB scheduler can be further simplified. For instance, let the parameters $a$ and $b$ satisfy that
\begin{equation}
g(a| \theta_{0},\theta_{1})=\frac{1}{\alpha} \quad\text{ and }\quad g(b| \theta_{0},\theta_{1})=\alpha.
\end{equation}
It follows that
\begin{equation*}
\begin{split}
1&=\frac{1}{\alpha}\cdot\alpha =g(a| \theta_{0},\theta_{1})\cdot g(b| \theta_{0},\theta_{1})\\
&=\exp({[(\theta _1 - \theta _0)/{\sigma ^2}]\cdot[{a+b-(\theta_{0}+\theta_{1})}]}),
\end{split}
\end{equation*}
which implies that
\begin{equation}
\label{Sumofab}
a+b=\theta_{0}+\theta_{1}.
\end{equation}
We also notice that
\begin{equation*}
\begin{split}
g(y|{\theta _0},{\theta _1}) &= \frac{{{p_{{\theta _1}}}(y)}}{{{p_{{\theta _0}}}(y)}} = \frac{{{\exp({ - {{(y - {\theta _1})}^2}/2{\sigma ^2}}}})}{{{\exp({ - {{(y - {\theta _0})}^2}/2{\sigma ^2}}}})}\\
&= {\exp({y({\theta _1} - {\theta _0})/{\sigma ^2}})}{\exp({(\theta _0^2 - \theta _1^2)/2{\sigma ^2}})}.
\end{split}
\end{equation*}
It is clear that $g(y|{\theta _0},{\theta _1})$ is an increasing function in $y$. Thus the LRB scheduler in (\ref{LRB scheduler}) can be explicitly rewritten as follows
\begin{equation}
\label{lrbscheduler}
\gamma_{i}=\left\{
\begin{array}{rcl}
0,       &      &\text{if  } {a<y_{i}<b};\\
1,     &      & \text{otherwise,}
\end{array} \right.
\end{equation}
where $a<b$ and $a+b=\theta_{0}+\theta_{1}$. See Fig. 2 for an illustration, which indicates that $y_i$ close to the center of $\theta_0$ and $\theta_1$ is not useful for the detection problem. This is consistent with our intuition as it is usually difficult to discriminate whether $y_i$ is sampled under $\cH_0$ or $\cH_1$ when $y_i$ is close to the center of $\theta_0$ and $\theta_1$. From this point of view, it explains the motivation of developing the LRB scheduler.

\begin{figure}[htbp!]
  \centering
  \includegraphics[width=3.4in]{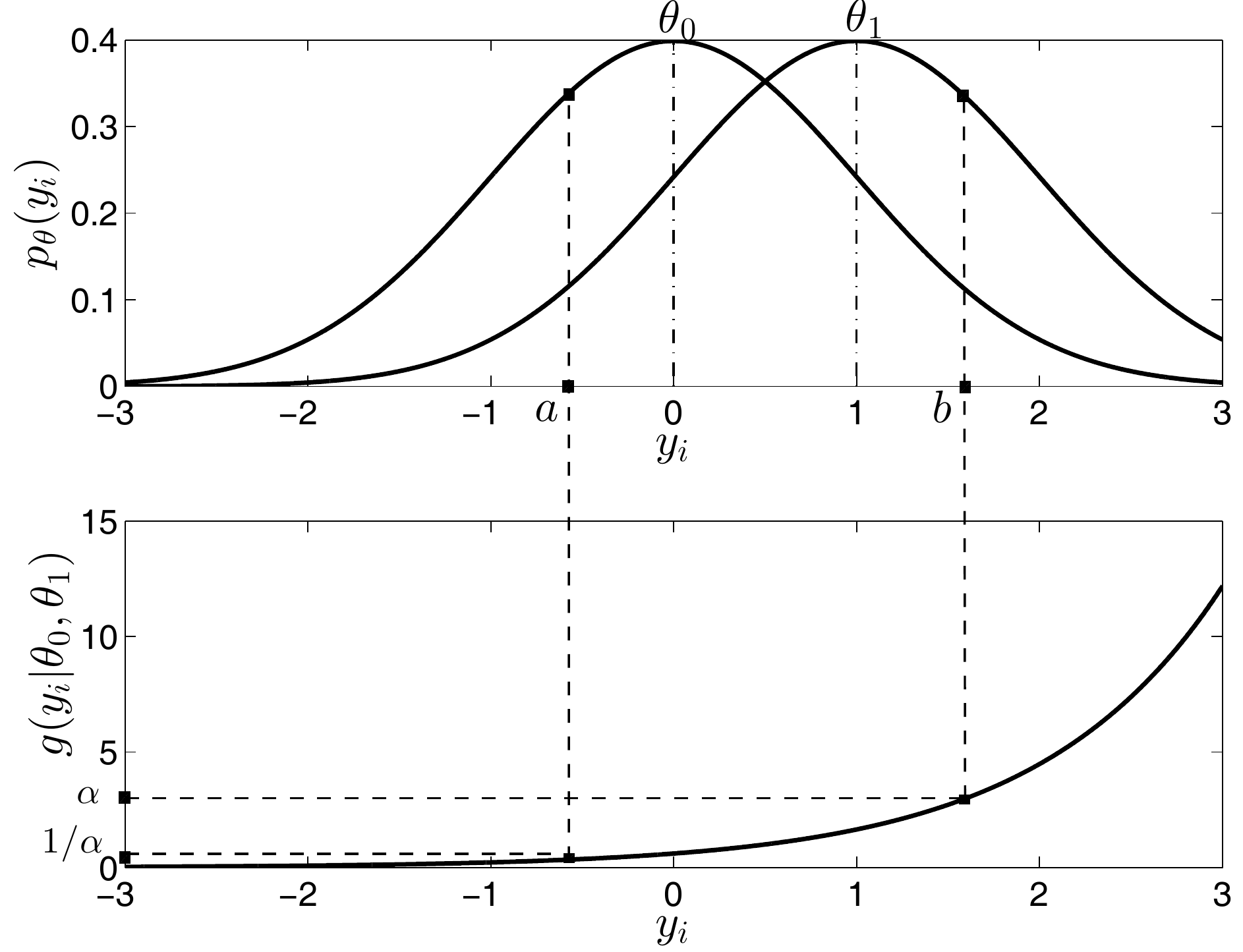}
    \caption{Scheduler thresholds}
  \label{fig2}
\end{figure}

\subsection{Random Scheduler}
If the sensor is not smart enough, it randomly selects a measurement and transmits it to the tester. To investigate this situation, we also consider a random scheduler for the detection problem by letting
\begin{equation}\label{randomscheduler}
{\gamma _i} = \left\{ \begin{array}{l}
1,{\quad\text{with~a~probability}}\;{p;}\\
0,{\quad\text{with~a~probability}} \;1 - {p},
\end{array} \right.
\end{equation}
and the scheduling sequence $\gamma_i$ is an i.i.d. process, which is independent of sensor measurements.

 In comparison with the random scheduler, $\gamma_i$ of the LRB scheduler is a deterministic function of the sensor measurement $y_i$.  If $y_i$ is not transmitted by the sensor, the tester deduces that $a<y_i<b$, which can also be utilized to improve the detection performance. Thus, it is expected that the LRB scheduler may achieve a better detection performance under the same scheduled transmission rate. We shall formally confirm it in the sequel.

\section{Performance Analysis under Normal Conditions}
\label{sec_npt}

In this section, we focus on the LRB scheduler design under secure communication channel, while the case under cyber attacks  is discussed in Section \ref{section4}.
\subsection{Neyman-Pearson Test}
We explicitly establish an N-P test under the scheduled communication in secure communication environment.  To this end, define a rejection region of the null hypothesis for the N-P test as
\begin{equation*}
\cR=\{\textbf{z}:f(\textbf{z}|\theta_{1})>kf(\textbf{z}|\theta_{0})\},
\end{equation*}
and an
acceptance region of $\cH_0$ in favor of $\cH_1$ as
\begin{equation*}
\cR^{c}=\{\textbf{z}:f(\textbf{z}|\theta_{1})<kf(\textbf{z}|\theta_{0})\},
\end{equation*}
where $f(\textbf{z}|\theta)$ is the joint pdf\footnote{Strictly speaking, it is not very accurate to use the pdf of ${\bf z}$ due to that the random vector ${\bf z}$ is evidently not continuous. To capture the essence of  our idea,  we state the pdf of ${\bf z}$ throughout this paper rather than adopt the probability measure.} of $\textbf{z}:=\{z_{1},\ldots, z_{N}\}$ corresponding to $\theta\in\{\theta_0,\theta_1\}$.

It follows from Neyman-Pearson lemma \cite{casella2001statistical} that the above test is a uniformly most powerful (UMP) test. That is, if we restrict to the class of all level $\alpha$ tests (control the probability of Type I error as at most $\alpha$), the N-P test would lead to the smallest probability of Type II error.  Note that a Type I error (or error of the first kind) is the incorrect rejection of a true null hypothesis, and the probability of the Type I error is conceptually given by $P_{\theta_0}(\cR)$. While a Type II error (or error of the second kind) is the failure to reject a false null hypothesis, and its probability is denoted by $P_{\theta_1}(\cR^c)$.

\subsubsection{N-P Test Under a LRB Scheduler}
Under the LRB scheduler (\ref{lrbscheduler}), the information received by the tester at time $N$ is written as  ${\bf{z}} = \{ {z_1}, \ldots ,{z_N}\}$, whose pdf  is given by
\beq
\label{pdf}
f(\textbf{z}|\theta)&=&\prod\limits_{i = 1}^N {{p_\theta }({z_i})} \nonumber\\
& =& \prod\limits_{i = 1}^N {{{\left[{{p_{T}}({y_i}|\theta ,\sigma^2 ,a,b)}  \right]}^{{\gamma _i}}}{{\left[ {{P}_\theta\{ {\gamma _i} = 0\} } \right]}^{1 - {\gamma _i}}}},\quad
\enq
where $p_{T}(\cdot)$ is the function of the truncated Gaussian pdf that lies within the set $U= ( - \infty ,a] \cup [b,\infty )$, and is  expressed as
\begin{equation}
\label{pT}
{p_{T}}({y}|\theta ,\sigma^2 ,a,b) = {{\frac{1}{\sigma }\varphi \left(\frac{{{y} - \theta }}{\sigma }\right)}}{ \cdot I_{U}(y)},
\end{equation}
where $I_{U}(y)$ is an indicator function, and $\varphi(\cdot)$ is the pdf of a standard Gaussian random variable, i.e.,
\begin{equation*}
\begin{split}
\varphi(t)&=\frac{1}{\sqrt{2\pi}}\exp(-t^{2}/2);\\
{I_U}(y) &= \left\{ \begin{array}{l}
1,{\quad\rm{ if}}\; y \in U,\\
0,{\quad\rm{ else}}.
\end{array} \right.
\end{split}
\end{equation*}


Then, the null hypothesis will be rejected if the following condition is satisfied
\begin{equation}
\label{rejectregion}
\prod_{i=1}^N \left(\frac{p_T(y_i|\theta_1,\sigma^2,a,b)}{p_T(y_i|\theta_0,\sigma^2,a,b)}\right)^{\gamma_i}\left(\frac{P_{\theta_1}\{\gamma_i=0\}}{P_{\theta_0}\{\gamma_i=0\}}\right)^{1-\gamma_i}> k
\end{equation}
where $k\in\bR$ is to be designed to satisfy a given performance requirement. To simplify the notation, let
\begin{equation*}
\begin{gathered}
\Phi(x)=\int_{-\infty}^{x}\varphi(t)dt;\hfill\\
z^1(\theta ) = \frac{{{a} - \theta }}{\sigma },
z^2(\theta ) = \frac{{{b} - \theta }}{\sigma }.\hfill\\
\end{gathered}
\end{equation*}
The probability of event $\{\gamma_{i}=0\}$ can be easily obtained as
\begin{equation}
\label{Pgamma}
{P}_{\theta}\{\gamma_{i}=0\}=\Phi(z^2(\theta ))-\Phi(z^1(\theta )).
\end{equation}
By (\ref{pT}) - (\ref{Pgamma}), we shall reject $\cH_{0}$ if
\begin{eqnarray}
\label{rejectlrb}
&&  \frac{{ - (\theta _1^2 - \theta _0^2)}}{{2{\sigma ^2}}} \cdot \sum\limits_{i = 1}^N {{\gamma_{i}}}  + \frac{{({\theta _1} - {\theta _0})}}{{{\sigma ^2}}}\sum\limits_{i = 1}^N {{\gamma_{i}y_i}}  \nonumber \\
&& >\ln k + (N - \sum\limits_{i = 1}^N {{\gamma _i}} )\ln\left( {\frac{{\Phi (z^2({\theta _0})) - \Phi (z^1({\theta _0}))}}{{\Phi (z^2({\theta _1})) - \Phi (z^1({\theta _1}))}}} \right)\nonumber\\
&&  =\ln k,
\end{eqnarray}
where $\ln(\cdot)$ is the natural logarithm function, and the second term in the last inequality becomes zero.

\subsubsection{N-P Test under a Random Scheduler}
The joint pdf of ${\bf{z}} = \{ {z_1}, \ldots, {z_N}\}$ with respect to $\theta$ under a random scheduler (\ref{randomscheduler}) is given by
\begin{eqnarray}
f({\bf{z}}| \theta ) = \prod\limits_{i = 1}^N {{p_\theta }({z_i})}  = \prod\limits_{i = 1}^N {{{\left[ {p \cdot {\cal N}(y_{i}|\theta ,\sigma^2 )} \right]}^{{\gamma _i}}}{{\left[ {1 - p} \right]}^{1 - {\gamma _i}}}}.
\end{eqnarray}
Similarly, we shall reject $\cH_{0}$ in favor of $\cH_1$ if the following condition is satisfied
\begin{equation}
\label{rejectrandom}
  \frac{{ - (\theta _1^2 - \theta _0^2)}}{{2{\sigma ^2}}} \cdot \sum\limits_{i = 1}^N {{\gamma_{i}}}  + \frac{{({\theta _1} - {\theta _0})}}{{{\sigma ^2}}}\sum\limits_{i = 1}^N {{\gamma_{i}y_i}}  >  \ln k.
\end{equation}
\begin{rem}
In light of (\ref{rejectlrb}) and (\ref{rejectrandom}), it is interesting that the rejection conditions of both schedulers are of the same form. This means that there is no need for the tester to know which type of scheduler is used in the sensor node. Note that $\gamma_i$ in (\ref{rejectlrb}) is completely determined by $y_i$ while it is independent of ${\bf y}$ in  (\ref{rejectrandom}). This will result in different detection performance.

For the LRB scheduler in  (\ref{LRB scheduler}),  it may not be possible to explicitly calculate the probabilities of Type I and II errors  under a given parameter $k$, although numerical algorithms can be used.  However, both probabilities will decrease exponentially to zero as the number of sensor measurements tends to infinity.  In the sequel, we shall quantify this convergence speed.
\end{rem}
\subsection{Asymptotic Detection Performance}
\label{sec_atp}
We  design an optimal LRB scheduler in the sense of achieving the fastest convergence of the probability of Type II error under a scheduled transmission rate constraint, and that the probability of Type I error is controlled below a certain level. Note that it is impossible to simultaneously minimize the probabilities of both Type I and II errors \cite{casella2001statistical}.

\subsubsection{Relative Entropy and Stein's Lemma}
We recall the Stein's lemma \cite{cover2006eit} that the best error exponent for the exponential convergence of the probability of Type II error is determined by the relative entropy (Kullback-Leibler divergence) from the null to alternative hypotheses.

 Given two probability measures $P_1$ and $P_0$ over a set $\cX$, and $
P_0$ is absolutely continuous with respect to $P_1$, the relative entropy $D$ from $P_0$ to $P_1$ is defined by
\begin{equation}
\label{kldistance}
D({P_0} || {P_1})=\int_\cX\ln\left(\frac{dP_0}{dP_1} \right)dP_0,
\end{equation}
where $dP_0/dP_1 $ is the Radon-Nikodym derivative \cite{ash2000pam} of $P_0$ with respect to $P_1$, provided that the expression on the right-hand side of (\ref{kldistance}) exists. In fact, $D(P_0||P_1)$ characterizes the distance between two probability measures, which is a key factor to deciding the asymptotic detection performance.

\begin{lem}(Stein's lemma) \label{lem_chernoff}Let $X_{1},X_{2},\ldots,X_{n}$ be an i.i.d. random sequence under the probability measure $Q$ over a set $\cX$. Consider a binary hypothesis testing problem
\bee
\label{hypotest}
\cH_{0}:~Q=P_0~\text{ versus}~ ~\cH_{1}:~Q=P_1,
\ene where $D(P_{0} \lVert  P_{1} )<\infty$. Let $\cR_{n}\subseteq \cX^{n}$ be a rejection region of the null hypothesis $\cH_{0}$ obtained by the N-P test. Let the probabilities of Type I and II errors be
\begin{equation}
\alpha_{n}=P_{0}^{n}(\cR_{n})~\text{and}~ \beta_{n}=P_{1}^{n}(\cR_{n}^c),
\end{equation}
where $$P_{j}^{n}(\cdot)=\underbrace{P_{j}\times \cdots \times P_{j}(\cdot)}_{n},$$ is the product measure \cite{ash2000pam}.

For $0<\epsilon<\frac{1}{2}$, define
\begin{equation}
\beta_{n}^{\epsilon}=\mathop {\min }\limits_{{\cR_n} \subseteq {\cX ^n}, {\alpha _n} < \epsilon } {\beta _n}.
\end{equation}
Then it follows that
\begin{equation}
\mathop {\lim }\limits_{n \to \infty } \frac{1}{n}\ln \beta _n^\epsilon  =  - D({P_0} \lVert {P_1}).
\end{equation}
\end{lem}

Lemma \ref{lem_chernoff}  quantifies the asymptotic performance of the N-P test. Specifically, if the probability of Type I error is controlled to be less than $\epsilon$, the number of samples required to guarantee the probability of Type II error below a certain level $\delta<1$ is approximately given by
\begin{equation}
n \approx \frac{1}{{D({P_0}\lVert{P_1})}}\ln \left( {\frac{1}{\delta }} \right).
\end{equation}

To achieve a given detection performance, the relative entropy thus can approximately evaluate the number of samples required for discriminating between the null and alternative hypotheses. The larger the relative entropy is, the smaller the number of samples is required to discriminate the null and alternative hypotheses by using the N-P test.

Hence it is sensible to maximize the relative entropy $D({P_0} || {P_1})$ to achieve an asymptotically optimal hypothesis test in the sense of minimizing the probability of the Type II error. We shall design the optimal LRB scheduler to maximize the relative entropy from $\cH_0$ to $\cH_1$.\\

\subsubsection{Optimal LRB Scheduler}
In view of (\ref{lrbscheduler}), the relative entropy for the binary detection problem (\ref{prb_test}) is a function of scheduling parameters $a$ and $b$. Thus, we rewrite $D({P_{\theta_0}}\lVert{P_{\theta_1}})$ under the scheduler  (\ref{lrbscheduler}) as
\begin{equation}\label{reentropy}
L_{1}(a,b)=D({P_{\theta_0}}\lVert{P_{\theta_1}}).
\end{equation}
\begin{defi}
Given a scheduled transmission rate constraint, i.e., $R_\theta\leqslant R$, the optimal LRB scheduler for the N-P test is solved by the following constrained optimization problem
\begin{eqnarray*}
&\text{maximize} & L_{1}(a,b)\\
&\text{subject to}& R_{\theta}\leqslant R ,\quad \theta\in\{\theta_0, \theta_1\},\\
&&a+b=\theta_0+\theta_1,\quad a\leqslant b.
\end{eqnarray*}
\end{defi}
Note that $R_\theta$ is introduced in (\ref{schedulerrate}). Since $y_i$ is an i.i.d. random sequence, it follows that $\gamma_i$ is also i.i.d. This implies that $R_\theta=\bE_\theta[\gamma_1]=P_\theta\{\gamma_1=1\}$.

The following theorem gives an optimal LRB scheduler and quantifies its asymptotic detection performance.

\begin{thm}\label{thm_optimal}
Given a scheduled transmission rate constraint, i.e., $R_\theta\leqslant R$, the relative entropy from $\cH_0$ to $\cH_1$ under the LRB scheduler (\ref{lrbscheduler}) is given by
\beq
L_{1}(a^{*},b^{*}) &=& \frac{{(\theta _1^2 - \theta _0^2)R}}{{2{\sigma ^2}}} - \frac{{{\theta _1} - {\theta _0}}}{{{\sigma ^2}}}\nonumber\\
&&\times\int_{( - \infty ,a^{*}] \cup [b^{*},\infty )}^{} {y{p_{{\theta _0}}}(y)}dy
\enq
where the scheduling thresholds $a^{*}$ and $b^{*}$ satisfy that
\begin{eqnarray}
\label{constraint11}
R_{\theta_0}&=& 1-\int_{\frac{{a^{*} - {\theta _0}}}{\sigma }}^{\frac{{b^{*} - {\theta _0}}}{\sigma }} {\varphi (t)dt}  =R,\\
\label{constraint12}R_{\theta_1}&=& 1-\int_{\frac{{a^{*} - {\theta _1}}}{\sigma }}^{\frac{{b^{*} - {\theta _1}}}{\sigma }} {\varphi (t)dt}  =R.
\end{eqnarray}
\end{thm}

\begin{proof}
See Appendix A.
\end{proof}

\begin{rem}
For the conventional detection problem using the full set of sensor measurements, which results in that $R=1$, it follows from Theorem \ref{thm_optimal} that $a^{*}=b^{*}$. This implies that
\begin{equation}
\begin{split}
L_{1}^{*}&= \frac{{(\theta _1^2 - \theta _0^2)}}{{2{\sigma ^2}}} - \frac{{{\theta _1} - {\theta _0}}}{{{\sigma ^2}}}\int_{ - \infty }^\infty  {y{p_{{\theta _0}}}(y)} dy\\
&= \frac{{{{({\theta _1} - {\theta _0})}^2}}}{{2{\sigma ^2}}}.\\
\end{split}
\end{equation}
\end{rem}

\subsubsection{Performance Analysis}
Now, we evaluate the asymptotic performance of the optimal LRB scheduler. Let us consider the random scheduler in (\ref{randomscheduler}).
\begin{lem}
\label{lem_rdnscheduler}
Given a scheduled transmission rate constraint, i.e., $R_\theta\leqslant R$, the maximum relative entropy under the random scheduler in (\ref{randomscheduler}) is given by
\begin{equation}
L_{2}=D({P_{\theta_0}}\lVert{P_{\theta_1}}) = \frac{{{{({\theta _1} - {\theta _0})}^2}}R}{{2{\sigma ^2}}}.
\end{equation}
\end{lem}

\begin{proof}
Recall that $z_i$ is an i.i.d. random sequence with the pdf given by
\begin{equation}
{p_{{\theta}}}(z) = {\left[ p\cdot{\mathcal{N}(y; {\theta },\sigma^2 )} \right]^{\gamma }}{\left[ {1 - p} \right]^{1 - \gamma }},
\end{equation}
where $ \theta\in\{\theta_0, \theta_1\}$ and $\gamma\in\{0, 1\}$.
Then, the relative entropy of the two distributions  from $\cH_0$ to $\cH_1$ under the random scheduler is computed as follows
\begin{equation*}
\begin{split}
&D({P_{\theta_0}}\lVert{P_{\theta_1}})\\
&= \int_{\{\gamma = 1\}} {\ln\frac{{{p_{{\theta _0}}}(z)}}{{{p_{{\theta _1}}}(z)}}d{F_{{\theta _0}}}(z)}  + \int_{\{\gamma  = 0\}} {\ln\frac{{{p_{{\theta _0}}}(z)}}{{{p_{{\theta _1}}}(z)}}d{F_{{\theta _0}}}(z)}\\
&= p \cdot \int_{ - \infty }^\infty  {{p_{{\theta _0}}}(y){\ln}\frac{{{p_{{\theta _0}}}(y)}}{{{p_{{\theta _1}}}(y)}}dy}  + 0\\
&= \frac{{(\theta _1^2 - \theta _0^2)p}}{{2{\sigma ^2}}} \int_{ - \infty }^\infty  {{p_{{\theta _0}}}(y)dy}  - \frac{{({\theta _1} - {\theta _0})}p}{{{\sigma ^2}}} \int_{ - \infty }^\infty  {y{p_{{\theta _0}}}(y)dy} \\
&=\frac{{{{({\theta _1} - {\theta _0})}^2}}p}{{2{\sigma ^2}}},
\end{split}
\end{equation*}
where $F_{\theta}(\cdot)$ is the distribution function of random variable ${z_i}$ under $\theta\in\{\theta_0, \theta_1\}$ and $p_{\theta}(y)$ is the pdf of $y$ with an abuse of notation.

Consider the scheduled transmission rate constraint, we obtain that
\begin{equation}
R_\theta=p\leqslant R,
\end{equation}
which establishes the desired result.
\end{proof}
In the following result, we prove that the LRB scheduler strictly improves the random scheduler.
\begin{thm}Under the same scheduled transmission rate constraint, i.e., $R_\theta\leqslant R$, the optimal LRB scheduler achieves a strictly better asymptotic N-P testing performance than the random scheduler if $0<R<1$, i.e.,
$$
L_1(a^*,b^*)>L_2$$
where $L_1(a^*,b^*)$ and $L_2$ are given in Theorem \ref{thm_optimal} and Lemma \ref{lem_rdnscheduler}, respectively.
\end{thm}
\begin{proof}
Under a scheduled transmission rate constraint $R<1$ and for any pair of $a^{*}$ and $b^{*}$ satisfying (\ref{constraint11}) and (\ref{constraint12}), it follows from Theorem \ref{thm_optimal} and Lemma \ref{lem_rdnscheduler} that
\begin{equation*}
\begin{split}
&L_{1}({a^*},{b^*}) - L_{2}\\
&= \frac{{({\theta _1} - {\theta _0})R}}{{{\sigma ^2}}} \cdot \int_{ - \infty }^\infty  {y{p_{{\theta _0}}}(y)dy} \\
&\quad\quad\quad\quad\quad- \frac{{({\theta _1} - {\theta _0})}}{{{\sigma ^2}}}\int_{ ( - \infty ,{a^*}] \cup [{b^*},\infty )}^{} {y{p_{{\theta _0}}}(y)} dy\\
& = \frac{{{\theta _1} - {\theta _0}}}{{{\sigma ^2}}}\left( {\int_{ a^* }^{b^*}  {y{p_{{\theta _0}}}(y)dy}  - (1 - R)\int_{ - \infty }^\infty  {y{p_{{\theta _0}}}(y)dy} } \right)\\
&= \frac{{{\theta _1} - {\theta _0}}}{{{\sigma ^2}}}\left( {\int_{a^{*}}^{b^{*}} {y{p_{{\theta _0}}}(y)dy}  - (1 - R){\theta _0}} \right).
\end{split}
\end{equation*}
Note that $a^*+b^*=\theta_0+\theta_1$ and $b^{*}>a^{*}$, we obtain that
\begin{equation*}
\theta_{0}-b^{*}<-|a^{*}-\theta_{0}|.
\end{equation*}
This further implies that
 \begin{equation*}
\begin{split}
&\int_{{a^*}}^{{b^*}} {y{p_{{\theta _0}}}(y)dy}  - (1 - R){\theta _0}\\
&= \int_{{a^*}}^{{b^*}} {(y - {\theta _0}){p_{{\theta _0}}}(y)dy}  + {\theta _0}\int_{{a^*}}^{{b^*}} {{p_{{\theta _0}}}(y)dy}  - (1 - R){\theta _0}\\
&=\sigma \int_{\frac{{{\theta _0} - {b^*}}}{\sigma }}^{\frac{{{b^*} - {\theta _0}}}{\sigma }} {t\varphi (t)dt}  - \sigma \int_{\frac{{{\theta _0} - {b^*}}}{\sigma }}^{\frac{{{a^*} - {\theta _0}}}{\sigma }} {t\varphi (t)dt}\\
& > 0 - \sigma \int_{ - \frac{{|{a^*} - {\theta _0}|}}{\sigma }}^{\frac{{|{a^*} - {\theta _0}|}}{\sigma }} {t\varphi (t)dt}  \\
&= 0.
\end{split}
\end{equation*}

Since $\theta_1>\theta_0$, it is obvious that
$
L_1({a^*},{b^*}) - L_2 >0,
$
which completes the proof. \end{proof}
\begin{exam}
\label{exam1}
We provide a graphical view in Fig. \ref{fig_comparison} to compare the error exponents (relative entropy) for the two schedulers under different scheduled transmission rate constraints by setting $\theta_{0}=0$, $\theta_{1}=1,$ and $\sigma^2=1$. For the random scheduler, the error exponent grows linearly with respect to the scheduled transmission rate. For the LRB scheduler, however, we observe a faster increase of error exponent  when transmission rate is lower than $0.4$. It is also noteworthy that when only $70\%$ of sensor measurements are sent, the detection performance of the LBR scheduler is only 2.1\% worse than the standard N-P test using the full set of measurements. This suggests that the use of the scheduler with the form of (\ref{lrbscheduler}) is an effective method for maintaining a good detection performance when the scheduled transmission rate is limited.

\begin{figure}
  \centering
  \includegraphics[width=3.45in,height=1.80in]{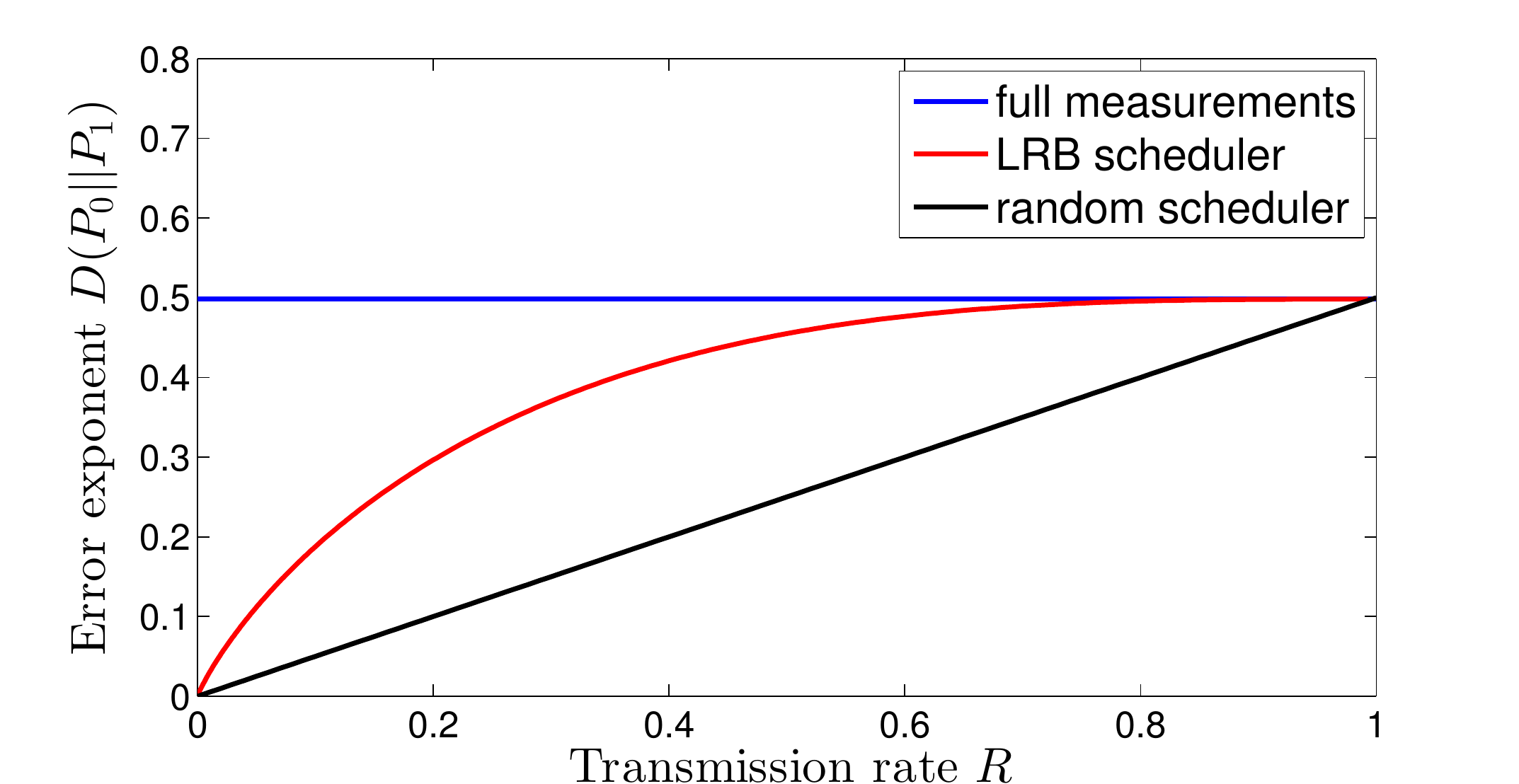}
    \caption{ Comparison of error exponents of the LRB scheduling strategy and random scheduling strategy under different scheduled transmission rate constraints.}
    \label{fig_comparison}
\end{figure}

\end{exam}

\section{Performance Analysis Under Cyber Attacks}
\label{section4}
Till now, we have considered the problem of networked detection problem via the Neyman-Pearson approach, where the communication channel is assumed secure. In this section, we studied the same problem over non-secure channels. In particular, the deceptive signals from the intruder are injected into the channel and arrive at the tester together with the true measurements. Then, if the sensor measurement $y_i$ is directly sent, it is impossible for the tester to distinguish the true measurement from the deceptive signal.

Observe that the binary state  $\gamma_i$  of the LRB scheduler can usually be securely encrypted, and successfully transmitted to the tester. We shall explore this information to retrieve some sensor measurements for accomplishing the detection task. To this purpose, we first propose a measurement discriminate protocol and then evaluate the error component for detection under different attack intensities.


\subsection{Measurement Discrimination}
Recall the LRB scheduler in (\ref{lrbscheduler}). If a packet with the message that $\gamma_i=0$ arrives,  it can be deduced by the tester that $y_i\in (a,b)$ and no sensor measurements shall be sent. When the message $\gamma_i=1$ comes, the tester is informed that a true sensor measurement should satisfy that $y_i\in(-\infty,a]\cup[b,\infty)$. Thus, we are able to design the following signal processing protocol for the tester to retrieve some sensor measurements while all the deceptive measurements are discarded.

\begin{itemize}
\item When $\gamma_i=0$ is observed by the tester, discard all the received measurements (if there is any).
\item When  $\gamma_i=1$ is observed by the tester and there is only one measurement  in the range $(-\infty,a]\cup[b,\infty)$, then select it as the true sensor measurement $y_i$. Otherwise, discard all the received data.
\end{itemize}

The first rule is obvious. Arrival of  $\gamma_i=0$ indicates that no sensor measurement is sent from sensors. Thus, all the received data are deceptive data and discarded. $\gamma_i=1$ suggests that a useful sensor measurement $y_i\in(-\infty,a]\cup[b,\infty)$ must have been transmitted to the tester. However, the tester may also receive a deceptive signal $q_i$. If $q_i\notin (-\infty,a]\cup[b,\infty)$, the tester can correctly identify the true measurement $y_i$. If $q_i\in (-\infty,a]\cup[b,\infty)$,  the tester has to discard all the received data since it is impossible to distinguish between $q_i$ and $y_i$.

If there is no scheduler in (\ref{lrbscheduler}), the tester can not distinguish between $y_i$ and $q_i$ at every sampling time, and thus fails to complete the detection task. From this point of view, the scheduler helps to against this type of cyber attacks. Obviously, the  detection performance degrades when deceptive signals come more frequently. To formalize it, let a binary variable $\zeta_i$ to denote whether the attacker injects a deceptive signal to the channel, e.g., $\zeta_i=1$ means that there is a deceptive signal $q_i$ in the channel, otherwise $\zeta_i=0$. Thus, {\it attack intensity} $P_i:=Pr\{\zeta_i=1\}$ reflects the frequency of attacking the networked system. We further assume that the deceptive signal $\{q_i\}$ forms an i.i.d. sequence, and let
\begin{equation}
\label{attackintensity}
\eta=1-P_i\cdot Pr\{q_i\in (-\infty,a]\cup[b,\infty)\}
\end{equation}

Under the above discrimination scheme,  each transmitted sensor measurement will be affected by cyber attacks with probability $1-\eta$.  In the rest of this section, we shall quantify the effect of cyber attacks on the detection performance by $\eta$.

\subsection{Performance Analysis}
The sensor information accessed by the tester at time $N$ is denoted by ${\bf{w}} = \{ {w_1}, \ldots ,{w_N}\}$, whose pdf is given by
\beq
\label{pdf_attack}
f(\textbf{w}|\theta)&=&\prod\limits_{i = 1}^N {{p_\theta }({w_i})} \nonumber\\
& =  &\prod\limits_{i = 1}^N {{{\left\{ {{{\left[ {\eta  \cdot {p_T}({y_i}|\theta ,{\sigma ^2},a,b)} \right]}^{\delta_i} }{{(1 - \eta )}^{1 - \delta_i}}} \right\}}^{{\gamma _i}}}} \nonumber\\
& & \cdot {\left[ {{P_\theta }\{ {\gamma _i} = 0\} } \right]^{1 - {\gamma _i}}},\quad
\enq
where $\delta_i$ is a binary random variable. It is one if the transmitted measurement $y_i$ is correctly identified by the tester and $\delta_i=0$ otherwise.



By a similar derivation in (\ref{pT})-(\ref{Pgamma}), we shall reject $\cH_{0}$ in favor of $\cH_{1}$ when the following condition is met.
\begin{equation}
\label{rejectRegion_attack}
\frac{{ - (\theta _1^2 - \theta _0^2)}}{{2{\sigma ^2}}} \cdot \sum\limits_{i = 1}^N {{\delta _i} \cdot {\gamma _i}}  + \frac{{({\theta _1} - {\theta _0})}}{{{\sigma ^2}}}\sum\limits_{i = 1}^N {{\delta _i} \cdot {\gamma _i}{y_i}}  > \ln k.
\end{equation}

\begin{rem}
It is noteworthy that $\delta_i$ equals to one if there is no cyber attack. The rejection region in (\ref{rejectRegion_attack}) is then reduced to the one in (\ref{rejectlrb}). This suggests that  the proposed detection framework is more general and can perform robustly under non-secure channels.
\end{rem}

The following theorem quantifies the influence of cyber attacks on the detection performance under different attack intensities.

\begin{thm}\label{thm_attack}
Under the LRB scheduler in (\ref{lrbscheduler}) with  cyber attacks, the relative entropy for detection is given by
\begin{equation}
D(P_{\theta_0}({w})||P_{\theta_1}({w}))=\eta \cdot D(P_{\theta_0}({z})||P_{\theta_1}({z})),
\end{equation}
where $D(P_{\theta_0}({z})||P_{\theta_1}({z}))$ denotes the relative entropy under secure channel and is given in (\ref{reentropy}), and $\eta$ is defined in (\ref{attackintensity}).
\end{thm}
\begin{proof}
Recall that $w_{i}$ is an i.i.d. random sequence with the pdf given by
\begin{equation*}
p_\theta(w) = {\left\{ {{{\left[ {\eta {p_T}(y|\theta ,{\sigma ^2},a,b)} \right]}^\delta }{{(1 - \eta )}^{1 - \delta }}} \right\}^\gamma }{\left[ {{P_\theta }\{ \gamma  = 0\} } \right]^{1 - \gamma }}
\end{equation*}
where $ \theta\in\{\theta_0, \theta_1\}$ and $\delta, \gamma\in\{0, 1\}$.
The relative entropy of the two distributions  from $\cH_0$ to $\cH_1$ under the sensor scheduler and  the proposed measurement discrimination over non-secure channels is computed as follows
\begin{equation*}
\begin{split}
&D(P_{\theta_0}({w})||P_{\theta_1}({w}))\\
&=\int_{\{ \gamma  = 1\} } {\ln \frac{{{p_{{\theta _0}}}(w)}}{{{p_{{\theta _1}}}(w)}}d{F_{{\theta _0}}}(w)}  + \int_{\{ \gamma  = 0\} } {\ln \frac{{{p_{{\theta _0}}}(w)}}{{{p_{{\theta _1}}}(w)}}d{F_{{\theta _0}}}(w)} \\
&=  \int_{\{ \gamma  = 0\} } {\ln \frac{{{p_{{\theta _0}}}(w)}}{{{p_{{\theta _1}}}(w)}}d{F_{{\theta _0}}}(w)} +
\int\limits_{\{ \gamma  = 1,\delta  = 1\} } {\ln \frac{{{p_{{\theta _0}}}(w)}}{{{p_{{\theta _1}}}(w)}}d{F_{{\theta _0}}}(w)}  \\
&~~~~+\int\limits_{\{ \gamma  = 1,\delta  = 0\} } {\ln \frac{{{p_{{\theta _0}}}(w)}}{{{p_{{\theta _1}}}(w)}}d{F_{{\theta _0}}}(w)} \\
&= \int_{\{ \gamma  = 1\} } {\eta  \cdot {p_T}(y|{\theta _0},\sigma ,a,b)\ln \frac{{{p_T}(y|{\theta _0},\sigma ,a,b)}}{{{p_T}(y|{\theta _1},\sigma ,a,b)}}dy}  + 0\\
&~~~~+ {P_{{\theta _0}}}\{ \gamma  = 0\}  \cdot \ln \frac{{{P_{{\theta _0}}}\{ \gamma  = 0\} }}{{{P_{{\theta _1}}}\{ \gamma  = 0\} }}\\
&= \eta  \cdot \int_{( - \infty ,a] \cup [b,\infty )}^{} {{p_{{\theta _0}}}(y)\ln \frac{{\exp [ - {{(y - {\theta _0})}^2}/2{\sigma ^2}]}}{{\exp [ - {{(y - {\theta _1})}^2}/2{\sigma ^2}]}}} dy\\
&~~~~+ {P_{{\theta _0}}}\{ \gamma  = 0\}  \cdot \ln \frac{{{P_{{\theta _0}}}\{ \gamma  = 0\} }}{{{P_{{\theta _1}}}\{ \gamma  = 0\} }}\\
&=  - \frac{{\eta({\theta _1} - {\theta _0}) }}{{{\sigma ^2}}}\left( {\int_{ - \infty }^a {y{p_{{\theta _0}}}(y)dy}  + \int_b^\infty  {y{p_{{\theta _0}}}(y)dy} } \right)\\
&~~~~+\frac{{\eta(\theta _1^2 - \theta _0^2) }}{{2{\sigma ^2}}}{P_{{\theta _0}}}\{ \gamma  = 1\} +{P_{{\theta _0}}}\{ \gamma  = 0\}  \cdot \ln \frac{{{P_{{\theta _0}}}\{ \gamma  = 0\} }}{{{P_{{\theta _1}}}\{ \gamma  = 0\} }}
\end{split}
\end{equation*}
From (\ref{Sumofab}), we obtain that
\begin{equation*}
\begin{split}
{P_{{\theta _0}}}\{\gamma  = 0\} &= \int_{\frac{{{a} - {\theta _0}}}{\sigma }}^{\frac{{{b} - {\theta _0}}}{\sigma }} {\varphi (t)dt}  = \int_{\frac{{{\theta _1} - {b}}}{\sigma }}^{\frac{{{\theta _1} - {a}}}{\sigma }} {\varphi (t)dt} \\
&= \int_{\frac{{{\theta _1} - {a}}}{\sigma }}^{\frac{{{\theta _1} - {b}}}{\sigma }} {\varphi ( - t)d( - t)}  = \int_{\frac{{{a} - {\theta _1}}}{\sigma }}^{\frac{{{b} - {\theta _1}}}{\sigma }} {\varphi (t)dt} \\
& = {P_{{\theta _1}}}\{\gamma  = 0\}.
\end{split}
\end{equation*}
It follows that
\begin{equation}
\begin{split}
&D(P_{\theta_0}({w})||P_{\theta_1}({w}))\\
&=  - \frac{{\eta({\theta _1} - {\theta _0}) }}{{{\sigma ^2}}}\left( {\int_{ - \infty }^a {y{p_{{\theta _0}}}(y)dy}  + \int_b^\infty  {y{p_{{\theta _0}}}(y)dy} } \right)\\
&~~~~+\frac{{\eta(\theta _1^2 - \theta _0^2) }}{{2{\sigma ^2}}}{P_{{\theta _0}}}\{ \gamma  = 1\}\\
&=\eta \cdot D(P_{\theta_0}({z})||P_{\theta_1}({z})),
\end{split}
\end{equation}
which completes the proof.
\end{proof}

By Lemma \ref{lem_chernoff}, an optimal LRB scheduler for detection should be designed to maximize the relative entropy from $\cH_0$ to $\cH_1$. Clearly,  $D(P_{\theta_0}({z})||P_{\theta_1}({z}))$ increases in $a$ but decreases in $b$. However, increasing $a$ will decrease the robustness of the LRB scheduler against cyber attacks, which is obvious by (\ref{attackintensity}) and Theorem \ref{thm_attack}. Note that the role of the parameter $b$ is similar to the parameter $a$. Together with Theorem \ref{thm_attack}, the optimal LRB scheduler for detection under cyber attacks can be defined.

\begin{defi}
Under cyber attacks, the optimal LRB scheduler in (\ref{lrbscheduler}) for the N-P test is solved by the following constrained optimization problem
\begin{eqnarray*}
&\text{maximize} & \eta \cdot D(P_{\theta_0}({z})||P_{\theta_1}({z}))\\
&\text{subject to}& R_{\theta}\leqslant R ,\quad \theta\in\{\theta_0, \theta_1\},\\
&&a+b=\theta_0+\theta_1,\quad a\leqslant b.
\end{eqnarray*}
\end{defi}

In view of the definition of $\eta$ in (\ref{attackintensity}), solving the above optimization problem requires the distribution of attack signal $q_i$, which however is generally unknown to the design of the LRB scheduler. Note that a typical detection system usually works under system state $\cH_0: \{\theta=\theta_0\}$ in the majority of time. An intelligent attacker is likely to fabricate a false alarm in the system by sending deceptive signals generated from distribution $P_{\theta_1}$ according to the alternative state $\cH_1:\{\theta=\theta_1\}$, i.e., $q_i\sim P_{\theta_1}$.

Under this case, we provide a graphical view in Fig. \ref{fig_comparison2} to illustrate how the error exponent (relative entropy) in Theorem \ref{thm_attack} changes under different attack intensities. Particularly, let $\theta_{0}=0$, $\theta_{1}=1,$ and $\sigma^2=1$ whereas the attack intensity varies from 0.1 to 1.0. The case when no attack happens (i.e., $P_i=0$) is also depicted for comparison.

\begin{figure}[htbp!]
  \centering
  \includegraphics[width=3.4in]{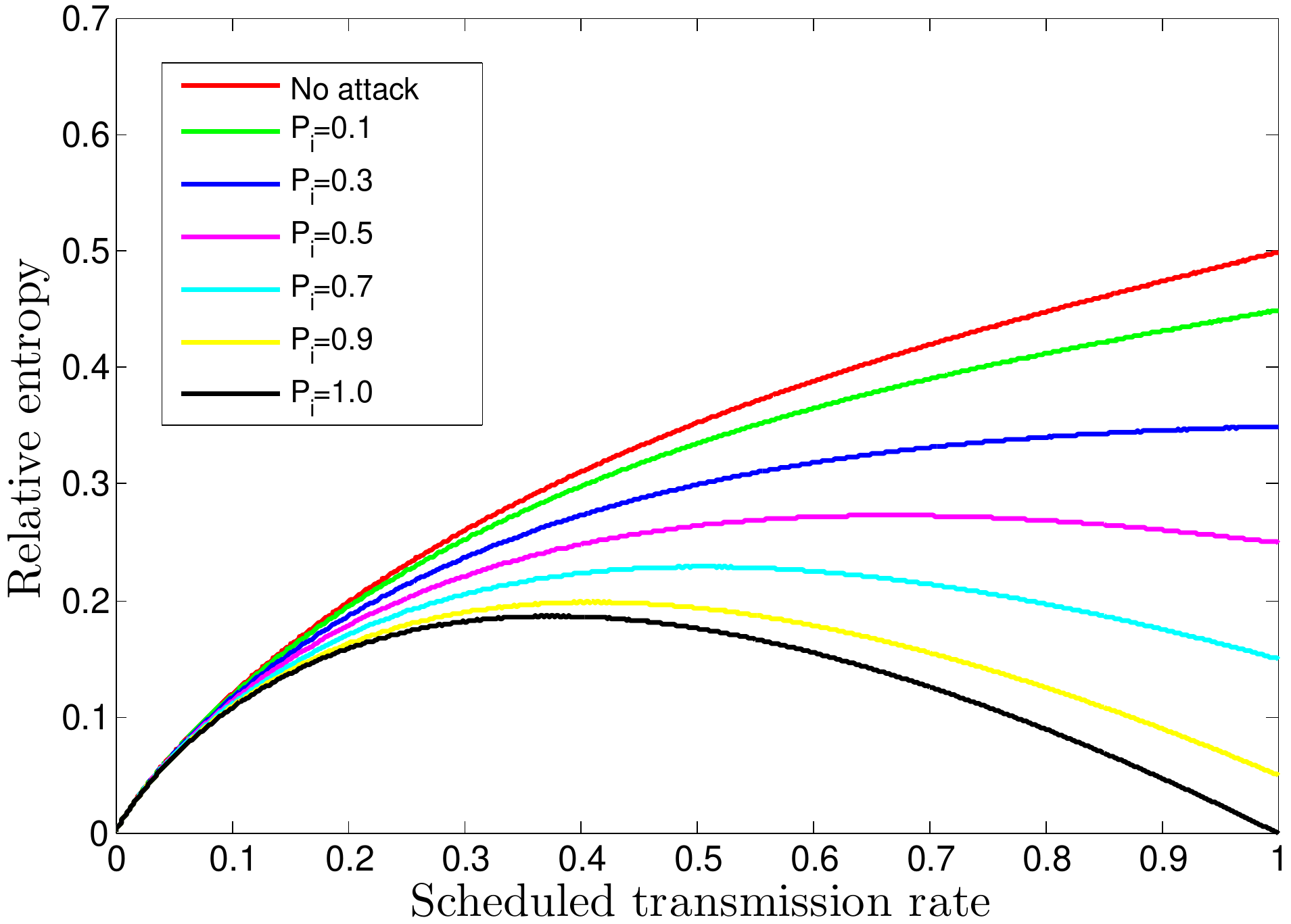}
    \caption{ Comparison of error exponents for detection under different attack intensities and different scheduled transmission rate constraints.}
    \label{fig_comparison2}
\end{figure}

By Fig. \ref{fig_comparison2}, it is apparent that a higher attack intensity results in a smaller error component, which is certainly consistent with the intuition as more deceptive signals will be received by the tester. It is interesting that the LRB scheduler with a lower scheduled transmission rate is more robust to cyber attacks. For example, when the rate is less than 0.3, the error exponents under different attack intensities are close. An instant explanation is that only a small portion of sensor measurements will be attacked by the deceptive signals.

The advantage of using the LRB scheduler becomes significant for a relatively large attack intensity ($P_i \ge 0.5$). Without the LRB scheduler, which corresponds to that the scheduled transmission rate is equal to one, the error components do not always attain their maximum at $R=1$. To get the best detection performance (e.g., the maximum error component), we have to resort to an effective LRB scheduler for the case that $P_i \ge 0.5$.

Under the extreme case ($P_i=1.0$), the error component even goes to zero as the scheduled transmission rate approaches one.  This implies that without  the LRB scheduler, the detection task can not be completed for the large attack intensity. In fact, without a LRB scheduler, every sensor measurement will be corrupted and can not be distinguished from the deceptive signals. Thus, all the sensor measurements will be discarded, leading to the failure of discriminating the system state.  This clearly explains the motivation for designing a LRB communication scheduler in the sensor side.

\section{Simulation}
In this section, the theoretical results will be validated respectively via simulations.

A deterministic signal in Gaussian noise is given by:
\begin{equation}
y_{i}=\theta+v_{i},i=1,2,\ldots,
\end{equation}
where $\theta\in \mathbb{R}$ is an unknown parameter for hypothesis testing, and $v_{i}$ is a white Gaussian noise with variance $\sigma^{2}=1$.

Consider a hypothesis testing problem as follows:
\begin{equation}
\cH_{0}:\{\theta=0\}\quad\quad  \text{versus} \quad\quad \cH_{1}:\{\theta=\theta_{1}\}.
\end{equation}

Obviously, the testing problem in this case is a signal detection problem, i.e. detecting the presence or absence of a deterministic signal in Gaussian enviornment. Accordingly, the signal to noise ratio (SNR) is given by
\begin{equation}
\text{SNR}=10\log_{10}\left( \frac{\theta_{1}^{2}}{\sigma^2}\right)\quad[\text{dB}].
\end{equation}
Low observation SNR of $-3$ dB,  medium observation SNR of $0$ dB and high observation SNR of $3$ dB are adopted in the numerical test.

The simulation is conducted from two aspects. Firstly, we verify the testing performance of the N-P test in Section \ref{sec_npt} under  the LRB scheduler with different communication constraints. For comparison, a random scheduler and an optimal N-P tester with the full set of measurements are considered. Secondly, we test the detection performance of the N-P test in Section \ref{section4} under the scheduler at the sensor and the  receiving protocols at the tester when false data injection attack happens.  Simulations are conducted under different transmission rates and attack intensities to show the robustness of the system in different scenarios.

\subsection{Detection Performance Without Cyber Attacks}

\begin{figure*}[htbp!]
\centering
\subfigure[Probability of Type II error at SNR -3dB] {\includegraphics[height=1.7in,width=2.35in]{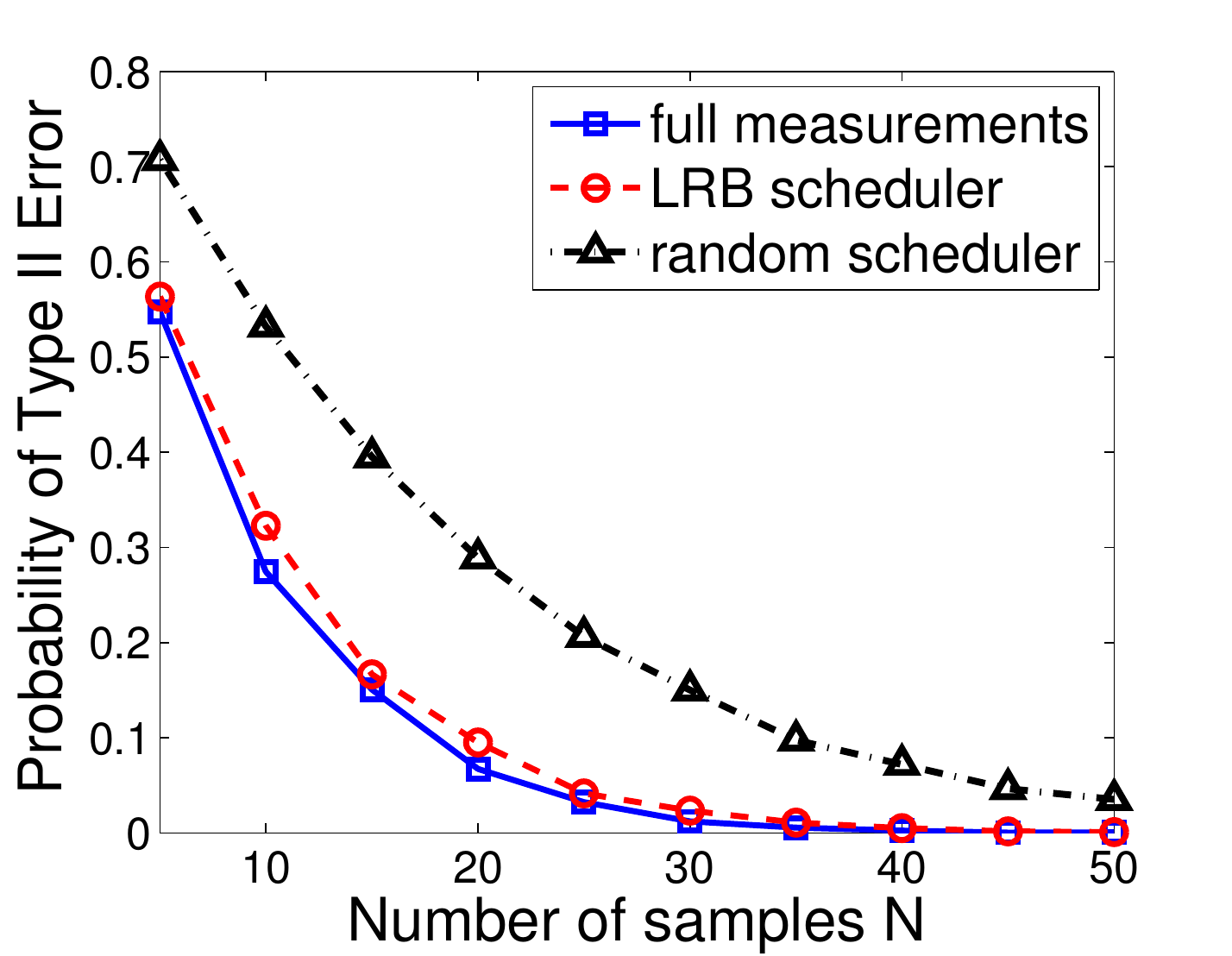}}
\subfigure[Probability of Type II error at SNR 0dB] {\includegraphics[height=1.7in,width=2.35in]{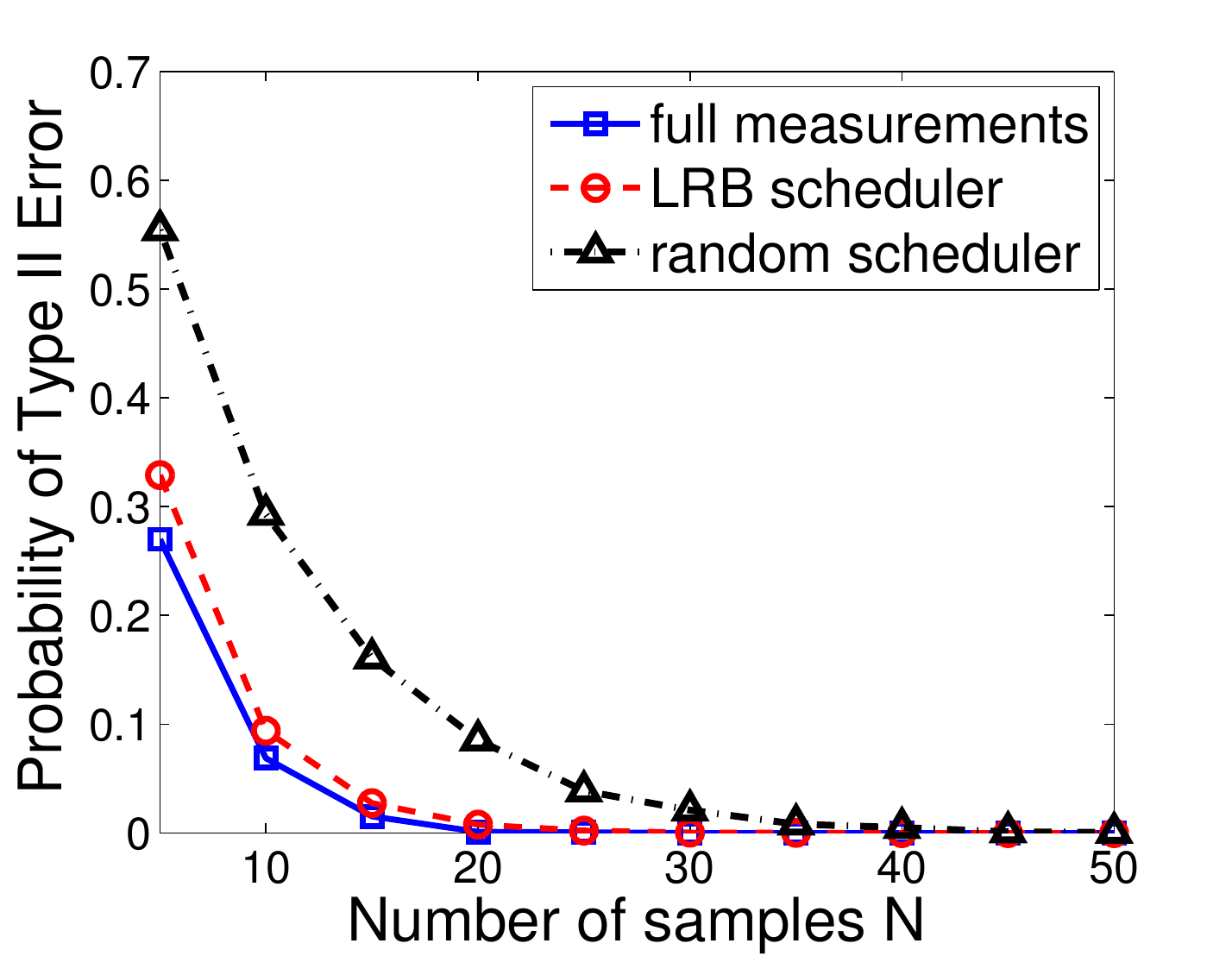}}
\subfigure[Probability of Type II error at SNR 3dB ] {\includegraphics[height=1.7in,width=2.35in]{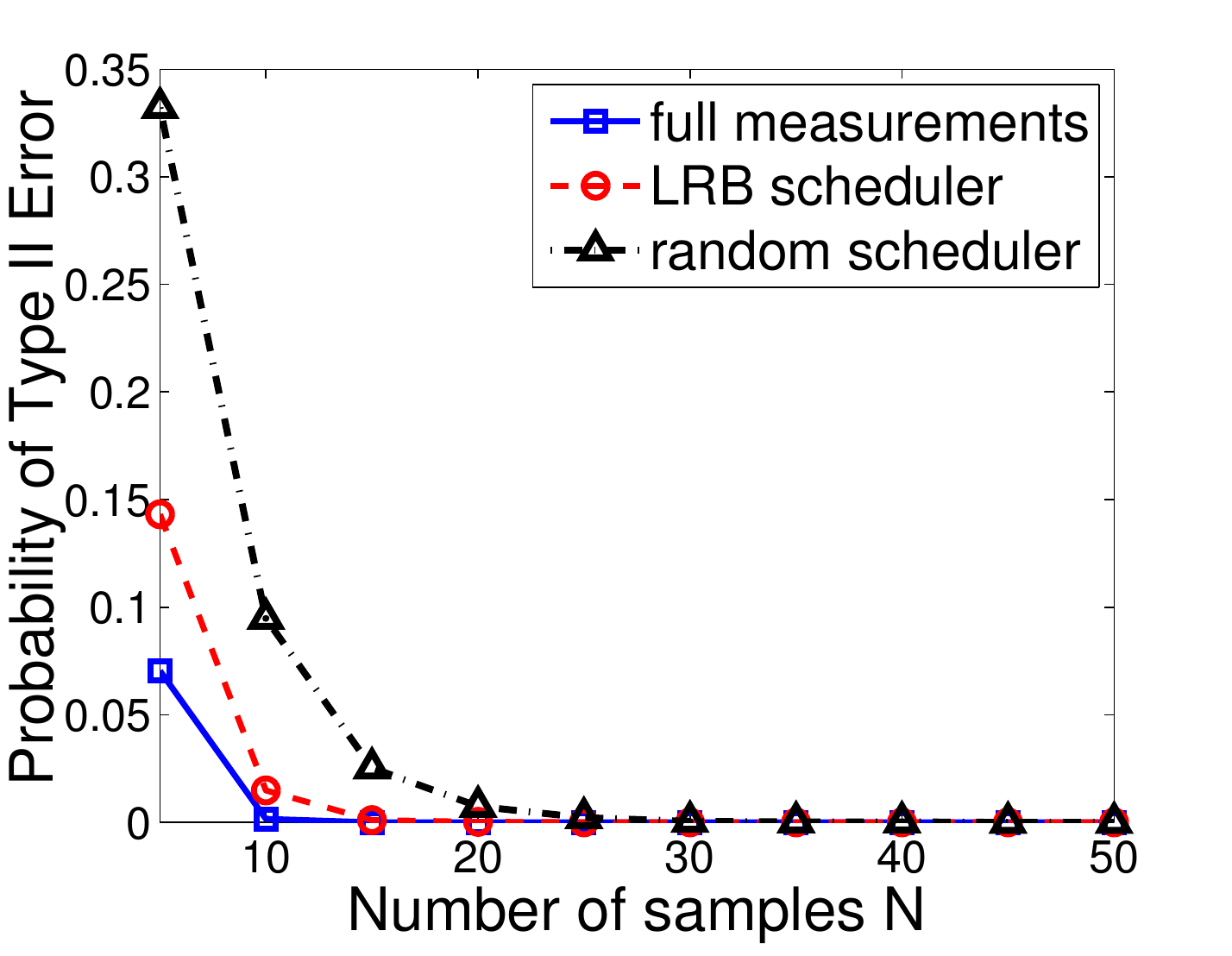}}
\caption{ Comparison of probability of Type II error under communication constraint $R_{\theta}\leqslant0.5$ in secure channel}
\label{NPsimulation1}
\end{figure*}

\begin{figure*}[htbp!]
\centering
\subfigure[Probability of Type II error at SNR -3dB] {\includegraphics[height=1.7in,width=2.35in]{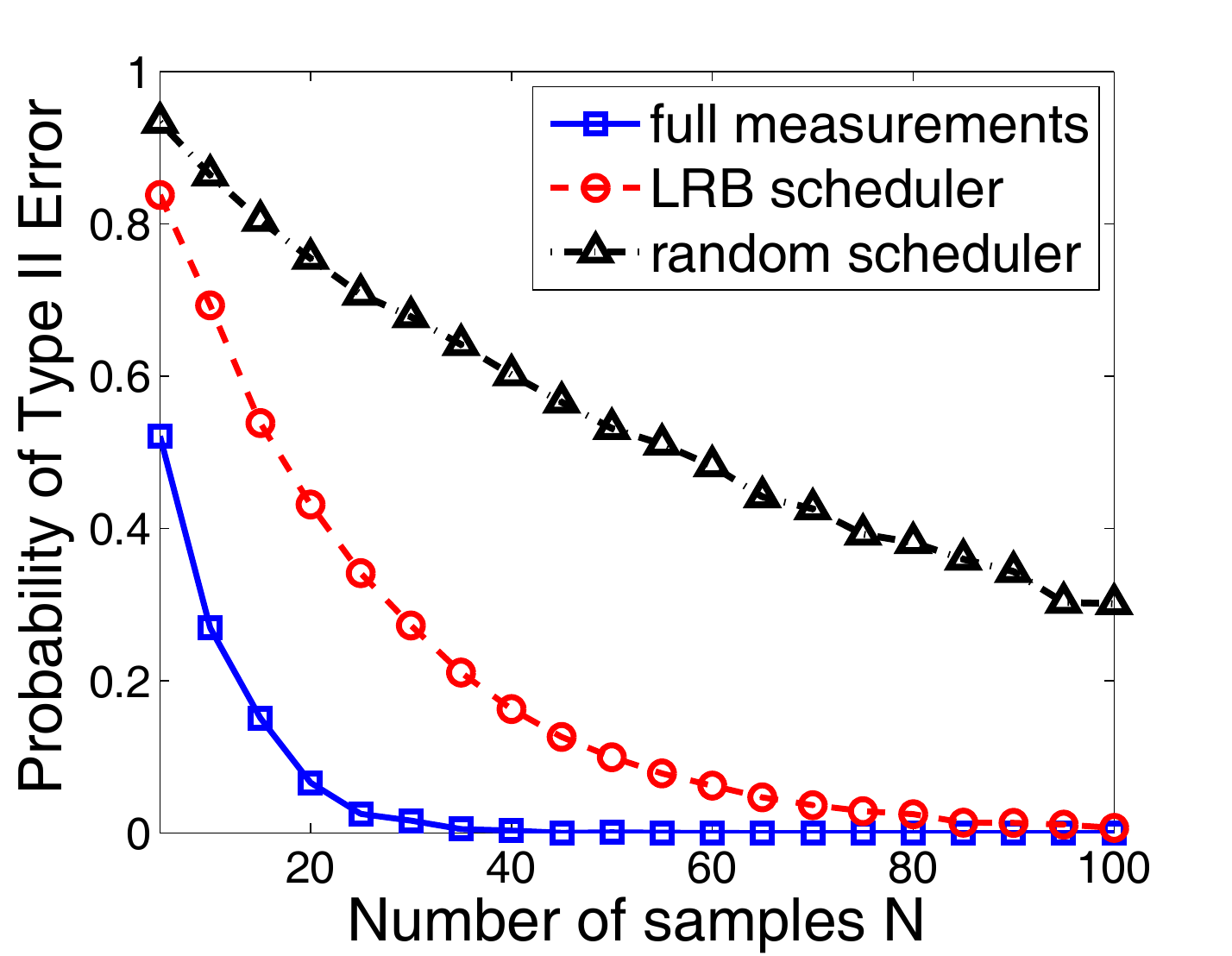}}
\subfigure[Probability of Type II error at SNR 0dB] {\includegraphics[height=1.7in,width=2.35in]{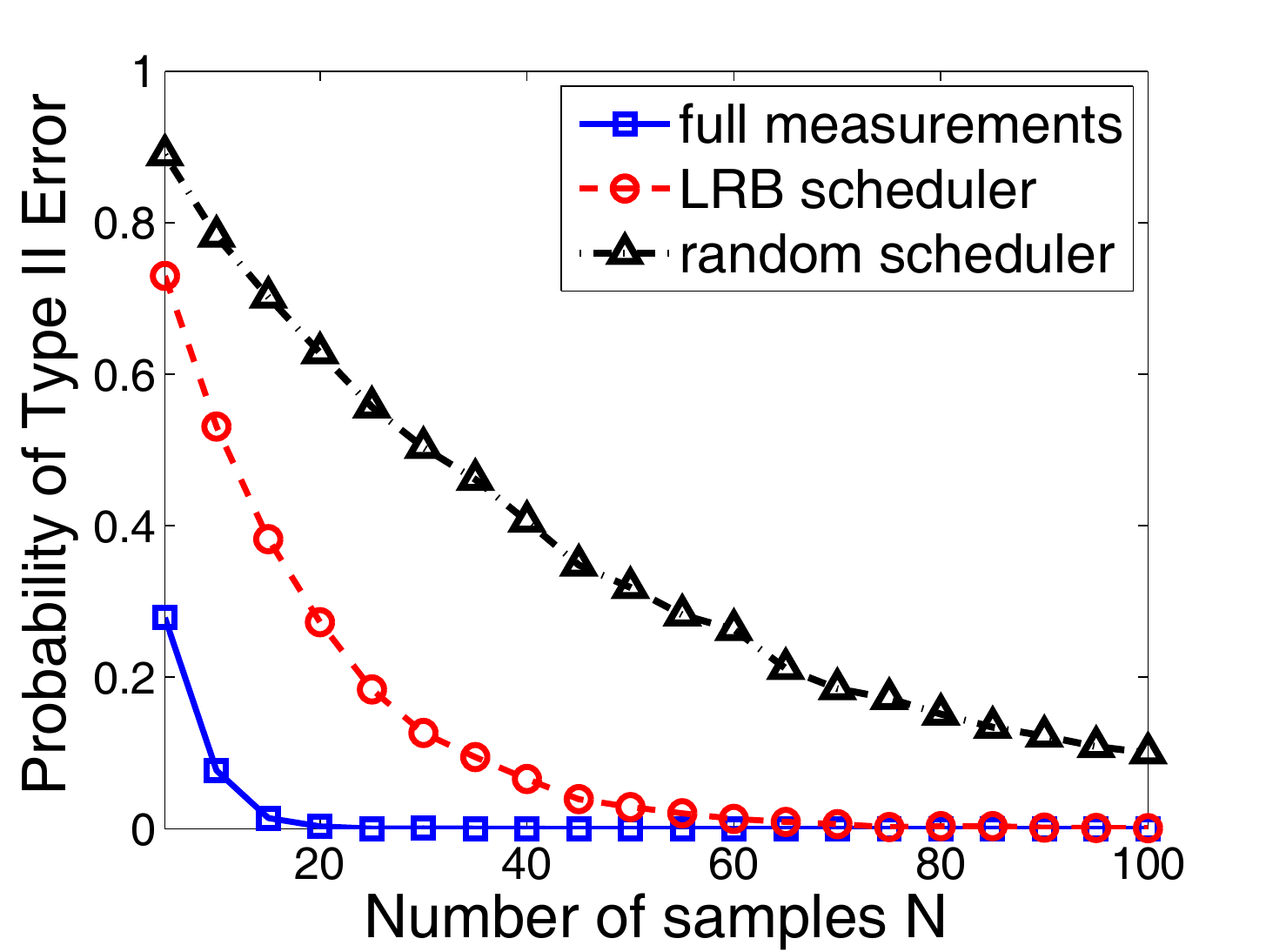}}
\subfigure[Probability of Type II error at SNR 3dB ] {\includegraphics[height=1.7in,width=2.35in]{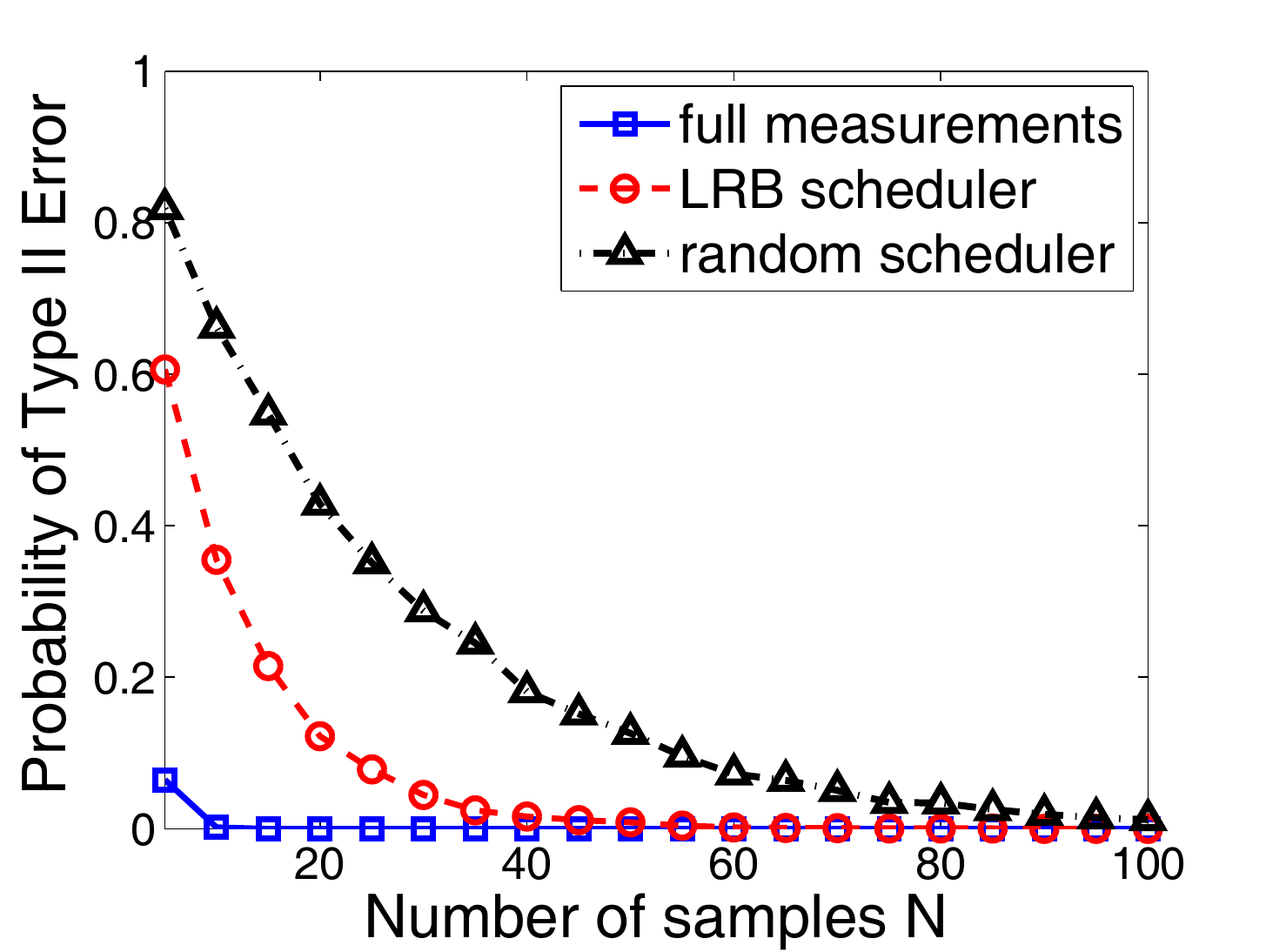}}
\caption{ Comparison of probability of Type II error under communication constraint $R_{\theta}\leqslant0.1$ in secure channel}
\label{NPsimulation2}
\end{figure*}

In this subsection, we check the testing performance of the scheduler based N-P test from the Neyman-Pearson approach, where we fix the probability of Type I error at $\alpha=0.05$ and numerically compute the probability of Type II error using the Monte Carlo method.  The channel is assumed to be secure, and the rejection condition of the null hypothesis is given by (\ref{rejectlrb}) and (\ref{rejectrandom}). Here we take a two-step procedure to perform the simulation.

In the first step, we determine the parameter $k$ to satisfy that the probability of Type I error is controlled to below $\alpha=0.05$ under different sizes of sensor measurements, and scheduled transmission rate constraints. Given a scheduled transmission rate constraint $R$, define
\beq\label{dataset}
T_N&:=&  \frac{{ - (\theta _1^2 - \theta _0^2)}}{{2{\sigma ^2}}} \cdot \sum\limits_{i = 1}^N {{\gamma_{i}}}  + \frac{{({\theta _1} - {\theta _0})}}{{{\sigma ^2}}}\sum\limits_{i = 1}^N {{\gamma_{i}y_i}}\nonumber\\
&=&T_{N-1}-\left(\frac{\theta _1^2}{2}\gamma_N-{\theta_1}\gamma_Ny_N\right),
\enq
where $\gamma_i$ is given in (\ref{lrbscheduler}) and Theorem \ref{thm_optimal}. For each $N$, we independently generated $S=5000$ samples from a Gaussian random variable under $\cH_0$. Then, we recursively obtain a data set with size $S$ for each $N$, i.e. $\cT_N=\{T_N^i\}_{i=1}^S$ using  (\ref{dataset}), from which an empirical distribution of $T_N$ can be obtained, and denoted as $F_N(\cdot, \omega)$ with $\omega\in\Omega$. By the Glivenko-Cantelli theorem \cite{billingsley2009convergence}, it is known that the empirical distribution converges  almost surely to the true distribution of $T_N$.  This means that $F_N(\cdot, \omega)$ is a good approximation of the true distribution of $T_N$. Then, the corresponding testing parameter $k$ for each $N$ is selected as
$$k_N=\exp(\inf_{x}\{x|F_N(x,\omega)\geqslant 1-\alpha\}).$$

In the second step,  the parameters determined above are used to compute the probability of Type II error. Similarly, we independently generated $S=5000$ samples from a Gaussian random variable under $\cH_1$, and recursively obtain a data set $\cT_N$ with size $S$ for each $N$. By the N-P test,  we count the number of samples that accept $\cH_0$, and the probability of Type II error is approximately given by the ratio of the total number of samples accepting $\cH_0$ to $S$, i.e.
$$\beta_N=\frac{\#\{i|T_N^i~\text{accepts}~\cH_0, 1\leqslant i\leqslant S\}}{S},$$
where $\#A$ denotes the cardinality of set $A$.

Fig. \ref{NPsimulation1} illustrates the numerically computed probability of Type II error under a scheduled transmission rate constraint $R_{\theta}\leqslant0.5$ while Fig. \ref{NPsimulation2} shows the result under a scheduled transmission rate constraint $R_{\theta}\leqslant0.1$. It is observed from Fig. \ref{NPsimulation1} that the probability of Type II error under the LRB scheduler comes close to the standard N-P test using the full set of measurements, and they are indistinguishable for a reasonable size of sensor measurements. Under both scheduled communication rate constrains, it is evident that the probability of Type II error of the LRB scheduler is smaller than the random scheduler. This numerically demonstrates the effectiveness of the LRB scheduler, and supports our theoretical results.

\subsection{Detection Performance Under Cyber Attacks}
In this subsection, we examine the detection performance of N-P test when deceptive signals are injected into the channel.
Similar to the previous subsection, the probability of Type I error is fixed at $\alpha=0.05$ and the probability of Type II error is obtained via the two-step Monte Carlo method with a modification on the testing statistic $T_N$ as follows.
\beq\label{dataset_attack}
T'_N&:=&  \frac{{ - (\theta _1^2 - \theta _0^2)}}{{2{\sigma ^2}}} \cdot \sum\limits_{i = 1}^N {{\delta_i\cdot \gamma_{i}}}  + \frac{{({\theta _1} - {\theta _0})}}{{{\sigma ^2}}}\sum\limits_{i = 1}^N {{\delta_i\cdot \gamma_{i}y_i}}\nonumber\\
&=&T'_{N-1}-\left(\frac{\theta _1^2}{2}\delta_N\gamma_N-{\theta_1}\delta_N\gamma_Ny_N\right).
\enq

The deceptive signals are assumed to be generated from distribution $P_{\theta_1}$ according to the alternative hypothesis $\cH_1:\{\theta=\theta_1\}$ to fabricate a false alarm in the system. The SNR is fixed at $0$ dB. Low attack intensity of $P_i=0.1$,  medium attack intensity of $P_i=0.5$ and high attack intensity of $P_i=1.0$ are considered in the simulation.

\begin{figure*}
\centering
\subfigure[Probability of Type II error at low attack intensity $P_i=0.1$] {\includegraphics[width=2.35in]{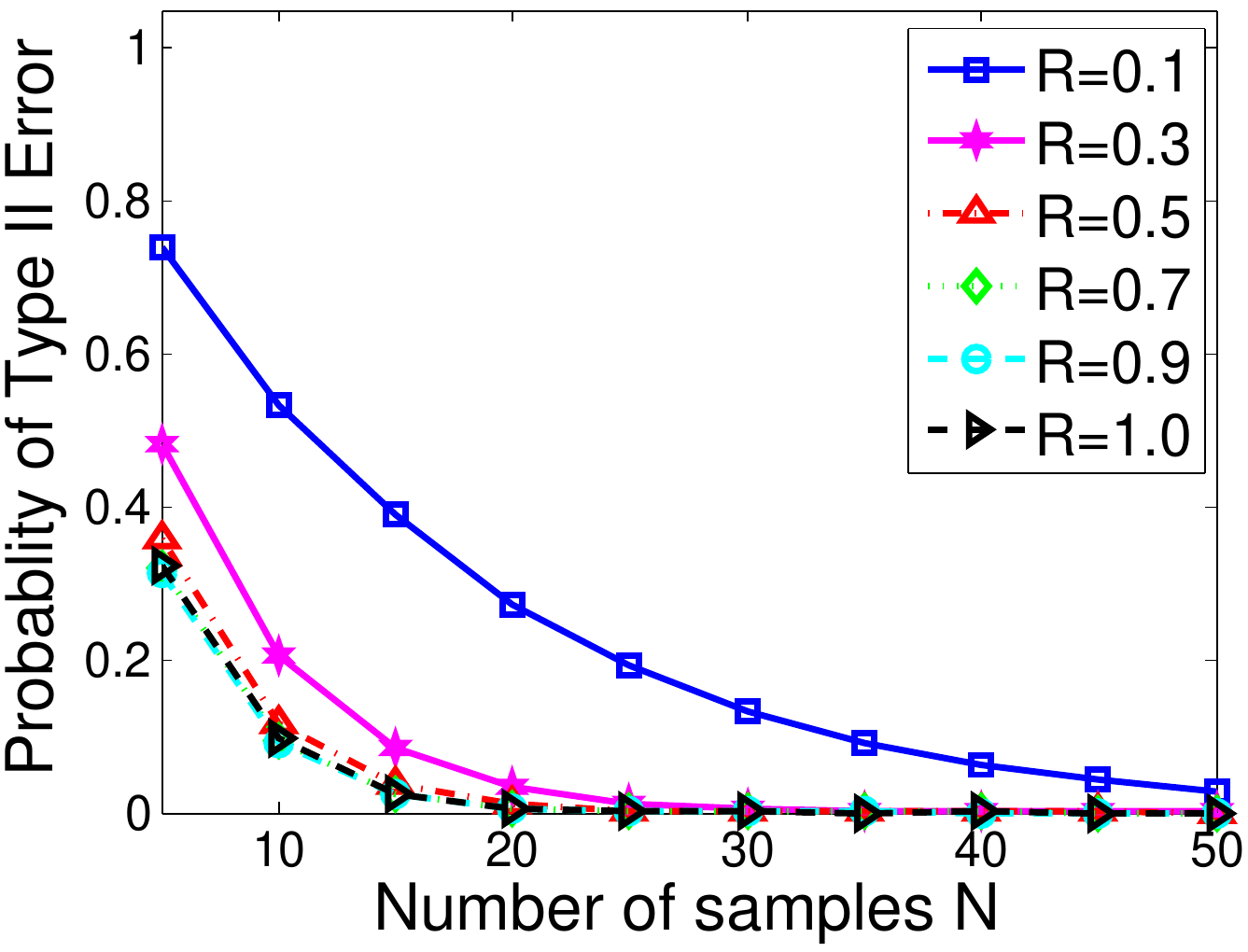}}
\subfigure[Probability of Type II error at medium attack intensity $P_i=0.5$] {\includegraphics[width=2.35in]{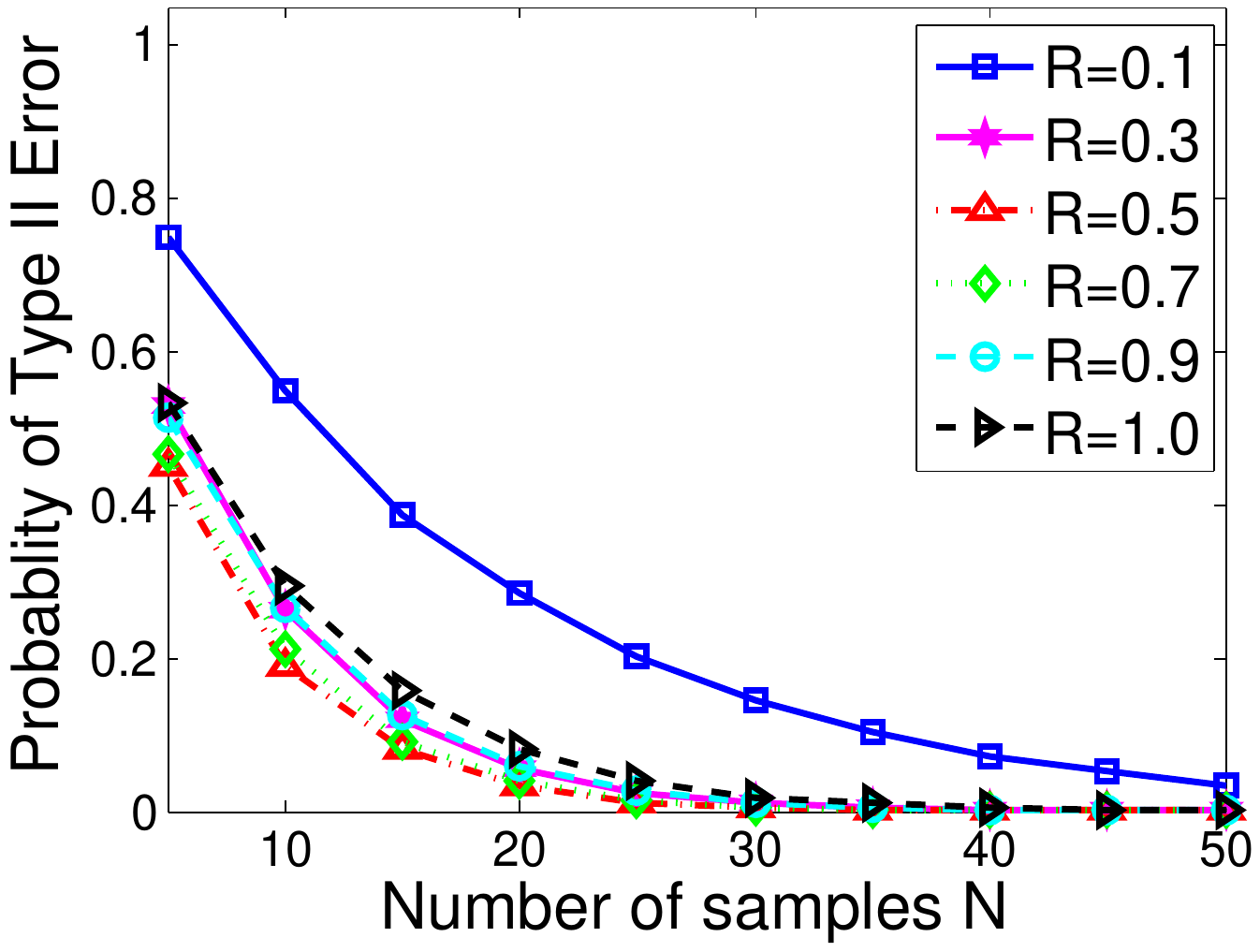}}
\subfigure[Probability of Type II error at high attack intensity $P_i=1.0$] {\includegraphics[width=2.35in]{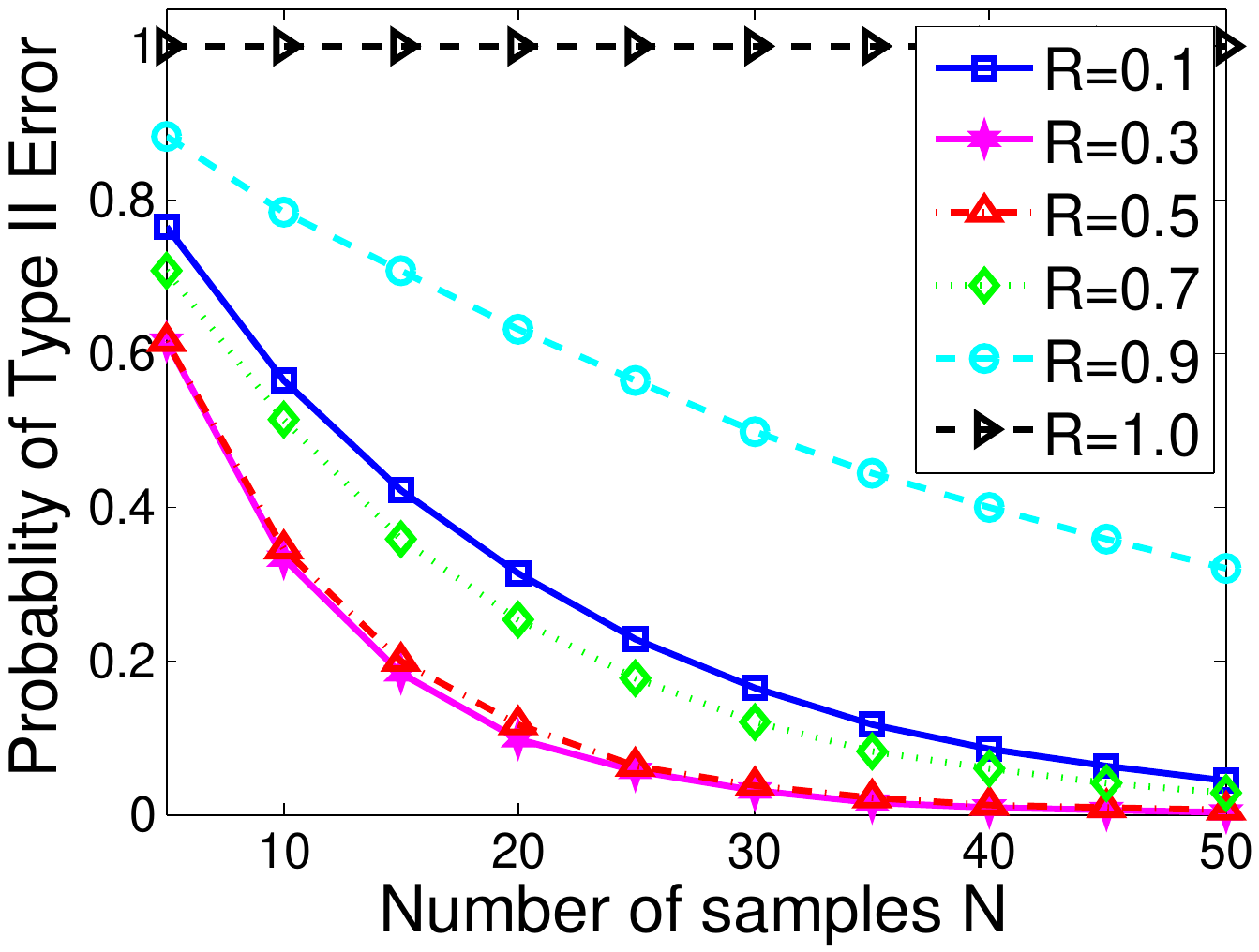}}
\caption{ Comparison of probability of Type II error under different attack intensities.}
\label{Attack_simulation}
\end{figure*}

Fig. \ref{Attack_simulation} depicts the numerically computed probability of Type II error when the system is being affected by cyber attacks. The detection performance under low, medium and high attack intensities are given by Fig. \ref{Attack_simulation} (a), (b), and (c) respectively. For a certain attack intensity, testing results under different levels of transmission rate constraints, e.g., $R\in\{0.1,0.3,0.5,0.7,0.9,1.0\}$ are illustrated in the figure.

It is clear from Fig \ref{Attack_simulation} (a) that a higher communication rate generally results in better detection performance under low attack intensities. For the medium and high attack intensities given in Fig. \ref{Attack_simulation} (b) and (c), however, the result is quite different. The best detection performance is obtained at $R=0.5$ and $R=0.3$ when the system is suffering cyber attacks with intensities $P_i=0.5$ and $P_i=1.0$ respectively. When the transmission rate approaches to 1, the testing performance deteriorates significantly. At high attack intensity ($P_i=1.0$), transmitting $90\%$ sensor measurements doesn't achieves a comparable performance to the case with merely $10\%$ measurements are transmitted. This scenario can be explained by the fact that a large portion of transmitted measurement are discarded at the tester since true sensor measurements become less and less discriminative among deceptive signals when the transmission rate increases.

\section{Conclusion}

Due to communication constraints and possible cyber attacks of a non-secure network, we have proposed a LRB scheduler for the binary detection problem.  The challenge of this testing framework lies in the nonlinearity of the scheduler.
Asymptotic analysis results demonstrate that the LRB scheduler for a secure network shares a comparable detection performance to the standard one using the full set of sensor measurements, even under a moderate communication cost. Meanwhile,  it achieves a strictly better exponent of  Type II probability of error than the random scheduler. Under cyber attacks by injecting deceptive signals to the network, the LRB scheduler can also maintain its detection functionality. The degradation of detection performance under different attack intensities was exactly quantified by means of asymptotic analysis.


Our future work is to extend the idea of the LRB scheduler to the multiple hypotheses case, i.e., $\theta\in\{\theta_0,\ldots,\theta_{M-1}\}$ and $M>2$. However, this entails a new challenge in selecting informative measurements, and conducting the asymptotic analysis. Another important future research is to study the vector measurement. There are two possible approaches for this generalization.  One is to sequentially decide the transmission of each element of the measurement vector.  The other is to compress the measurement vector into one dimensional scalar. Although these extensions will be interesting, they are beyond the scope of this paper and left to our future work.

\appendices
\section{Proof of Theorem \ref{thm_optimal}}
{\em Proof:}~Recall that $z_{i}$ is an i.i.d. random sequence with the pdf given by
\begin{equation*}
{p_{{\theta }}}(z) = {\left[ {{p_{T}}(y|\theta,\sigma ,a,b)} \right]^\gamma }{\left[ {P_{
\theta}\{{\gamma} = 0\}} \right]^{1 - \gamma }},
\end{equation*}
where $ \theta\in\{\theta_0, \theta_1\}$ and $\gamma\in\{0, 1\}$.
The relative entropy of the two distributions  from $\cH_0$ to $\cH_1$ under the LRB scheduling scheduler is computed as follows
\begin{equation*}
\begin{split}
&D({P_{\theta_0}}\lVert{P_{\theta_1}})\\
&=\int_{\{\gamma  = 1\}} {\ln\frac{{{p_{{\theta _0}}}(z)}}{{{p_{{\theta _1}}}(z)}}d{F_{{\theta _0}}}(z)}  + \int_{\{\gamma  = 0\}} {\ln\frac{{{p_{{\theta _0}}}(z)}}{{{p_{{\theta _1}}}(z)}}d{F_{{\theta _0}}}(z)} \\
&=\int_{\{\gamma  = 1\}} {{p_T}(y|{\theta _0},\sigma ,a,b)\ln\frac{{{p_T}(y|{\theta _0},\sigma ,a,b)}}{{{p_T}(y|{\theta _1},\sigma ,a,b)}}dy} \\
&\quad+ \int_{\{\gamma  = 0\}} {\ln\frac{{{P_{{\theta _0}}}\{\gamma  = 0\}}}{{{P_{{\theta _1}}}\{\gamma  = 0\}}}d{F_{{\theta _0}}}(y)}\\
&= \int_{( - \infty ,a] \cup [b,\infty )}^{} {{p_{{\theta _0}}}(y)\ln \frac{{\exp [ - {{(y - {\theta _0})}^2}/2{\sigma ^2}]}}{{\exp [ - {{(y - {\theta _1})}^2}/2{\sigma ^2}]}}}dy\\
&\quad + \ln\frac{{{P_{{\theta _0}}}\{\gamma  = 0\}}}{{{P_{{\theta _1}}}\{\gamma  = 0\}}} \cdot \int_a^b {{p_{{\theta _0}}}(y)} dy\\
&=\frac{{{P_{{\theta _0}}}\{\gamma  = 1\}\cdot(\theta _1^2 - \theta _0^2)}}{{2{\sigma ^2}}} +P_{{\theta _0}}\{\gamma=0\}\cdot\ln\frac{{{P_{{\theta _0}}}\{\gamma  = 0\}}}{{{P_{{\theta _1}}}\{\gamma  = 0\}}}\\
& \quad- \frac{{{\theta _1} - {\theta _0}}}{{{\sigma ^2}}}\left( {\int_{( - \infty ,a] \cup [b,\infty )}^{} {y{p_{\theta_{0}}}(y)dy} } \right).\\
\end{split}
\end{equation*}
From (\ref{Sumofab}), we obtain
\begin{equation}
\begin{split}
{P_{{\theta _0}}}\{\gamma  = 0\} &= \int_{\frac{{{a} - {\theta _0}}}{\sigma }}^{\frac{{{b} - {\theta _0}}}{\sigma }} {\varphi (t)dt}  = \int_{\frac{{{\theta _1} - {b}}}{\sigma }}^{\frac{{{\theta _1} - {a}}}{\sigma }} {\varphi (t)dt} \\
&= \int_{\frac{{{\theta _1} - {a}}}{\sigma }}^{\frac{{{\theta _1} - {b}}}{\sigma }} {\varphi ( - t)d( - t)}  = \int_{\frac{{{a} - {\theta _1}}}{\sigma }}^{\frac{{{b} - {\theta _1}}}{\sigma }} {\varphi (t)dt} \\
& = {P_{{\theta _1}}}\{\gamma  = 0\}.
\end{split}
\end{equation}
Then, it follows that
\begin{equation}
\label{Lab}
\begin{split}
L_1(a,b)&=D({P_{\theta_0}}\lVert{P_{\theta_1}})\\
& =  - \frac{{{\theta _1} - {\theta _0}}}{{{\sigma ^2}}}\left( {\int_{ - \infty }^a {y{p_{{\theta _0}}}(y)dy}  + \int_b^\infty  {y{p_{{\theta _0}}}(y)dy} } \right)\\
& \quad+ \frac{{(\theta _1^2 - \theta _0^2)}}{{2{\sigma ^2}}}{P_{{\theta _0}}}\{\gamma  = 1\}.
\end{split}
\end{equation}
Note that $a+b=\theta_0+\theta_1$ and $b>a$, we obtain
\begin{equation*}
|b-\theta_{0}|>|a-\theta_{0}|.
\end{equation*}
Therefore, it yields that
\begin{equation}
\begin{split}
&\frac{{\partial L_1(a,b)}}{{\partial a}}=  - \frac{{{\theta _1} - {\theta _0}}}{{{\sigma ^2}}}\left( {a{p_{{\theta _0}}}(a) + b{p_{{\theta _0}}}(b)} \right)\\
&\quad\quad\quad\quad\quad\quad\quad\quad\quad+ \frac{{(\theta _1^2 - \theta _0^2)}}{{2{\sigma ^2}}}\left( {{p_{{\theta _0}}}(a) + {p_{{\theta _0}}}(b)} \right)\\
&= \frac{{({\theta _1} - {\theta _0})(b - a)}}{{2{\sigma ^2}}} \cdot ({p_{{\theta _0}}}(a) - {p_{{\theta _0}}}(b))\\
&= \frac{{({\theta _1} - {\theta _0})(b - a)}}{{2{\sigma ^3}}} \cdot \left( {\varphi \left( {\frac{{a - {\theta _0}}}{\sigma }} \right) - \varphi \left( {\frac{{b - {\theta _0}}}{\sigma }} \right)} \right)\\
&>0,
\end{split}
\end{equation}
which implies that $L_1(a,b)$ is an increasing function of $a$.

Considering the constraints of the problem $R_{\theta}\leqslant R$, it follows that $a\leqslant a^*$. Therefore, $L_1(a,b)$ is maximized when $a=a^{*}$ and $b=b^{*}$.

By (\ref{Lab}), we obtain that:
\beq
L_{1}(a^{*},b^{*}) &=& \frac{{(\theta _1^2 - \theta _0^2)R}}{{2{\sigma ^2}}} - \frac{{{\theta _1} - {\theta _0}}}{{{\sigma ^2}}}\nonumber\\
&&\times\int_{( - \infty ,a^{*}] \cup [b^{*},\infty )}^{} {y{p_{{\theta _0}}}(y)}dy\nonumber.
\enq
The proof is completed.\qed

\bibliographystyle{IEEEtran}        
\bibliography{mybibf}           

\end{document}